\documentclass[10pt,reqno]{amsart}
\usepackage{amssymb,mathrsfs,color}
\usepackage{pinlabel}

\usepackage{amsmath, amsfonts, amsthm, verbatim, amssymb}
\usepackage{epstopdf, mathtools}
\mathtoolsset{showonlyrefs}
\usepackage{color}
\usepackage{esint}
\usepackage{stmaryrd}
\usepackage{graphicx}

\usepackage{empheq}
\usepackage[shortlabels]{enumitem}

\newcommand{\rar}{\rightarrow}

\newcommand{\cl}{\mathcal}

\def\grad {{\nabla}}

\newcommand{\bs}[1]{\boldsymbol{#1}}

\renewcommand{\div}{\operatorname{div}}

\definecolor{deepgreen}{cmyk}{1,0,1,0.5}

\newcommand{\R}{\mathbb{R}}

\newcommand{\al}{\alpha}
\newcommand{\be}{\beta}

\newcommand{\de}{\delta}

\newcommand{\la}{\lambda}

\newcommand{\p}{\partial}

\makeatletter

\newcommand{\Rmnum}[1]{\expandafter\@slowromancap\romannumeral #1@}
\makeatother

\newcommand{\Del}[1]{}

\usepackage[multiple]{footmisc}
\numberwithin{equation}{section}

\newtheorem{thm}{Theorem}[section]

\newtheorem{prop}[thm]{Proposition}

\theoremstyle{remark}

\usepackage{tensor}

\usepackage{graphicx}
\usepackage{xcolor} 
\usepackage{tensor}
\usepackage{slashed}
\usepackage{cite}
\definecolor{green}{rgb}{0,0.8,0} 

\newcommand{\tr}{\textrm{tr}\,}

\newcommand{\eps}{\epsilon}

\newcommand{\bbE}{\mathbb E}

\newcommand{\bbR}{\mathbb R}

\newcommand{\bbZ}{\mathbb Z}

\newcommand{\tens}{\otimes}

\newcommand{\Div}{\mbox{Div}\,}

\newcommand{\msf}{\bs{\mathsf{f}}}

\newcommand{\msa}{\bs{\mathsf{a}}}

\newcommand{\msg}{\bs{\mathsf{g}}}

\newcommand{\msu}{\bs{\mathsf{u}}}
\newcommand{\msM}{\bs{\mathsf{M}}}

\newcommand{\msK}{\bs{\mathsf{K}}}

\newcommand{\msE}{\bs{\mathsf{E}}}

\newcommand{\msT}{\bs{\mathsf{T}}}

\newcommand{\msP}{\bs{\mathsf{P}}}

\newcommand{\msd}{\bs{\mathsf{d}}}

\newcommand{\msk}{\bs{\mathsf{k}}}
\newcommand{\hbF}{\hat{\bs F}}
\newcommand{\hbA}{\hat{\bs A}}
\newcommand{\hbE}{\hat{\bs E}}
\newcommand{\hbP}{\hat{\bs P}}

\newcommand{\hbchi}{\hat{\bs \chi}}
\newcommand{\hbv}{\hat{\bs v}}
\newcommand{\sK}{\mathsf{K}}

\newcommand{\sab}{\mathsf{b}}

\newcommand{\sE}{\mathsf{E}}

\newcommand{\sg}{\mathsf{g}}
\newcommand{\sa}{\mathsf{a}}

\newcommand{\su}{\mathsf{u}}
\newcommand{\sv}{\mathsf{v}}

\begin{document}

\title[Gradient elastic material surfaces]{A midsurface elasticity model for a thin, nonlinear, gradient elastic plate}
\author{C. Rodriguez}

\begin{abstract} 
In this paper, we derive a dynamic surface elasticity model for the two-dimensional midsurface of a thin, three-dimensional, homogeneous, isotropic, nonlinear gradient elastic plate of thickness $h$. The resulting model is parameterized by five, conceivably measurable, physical properties of the plate, and the stored surface energy reduces to Koiter's plate energy in a singular limiting case. The model corrects a theoretical issue found in wave propagation in thin sheets and, when combined with the author's theory of Green elastic bodies possessing gradient elastic material boundary  surfaces, removes the singularities present in fracture within traditional/classical models. Our approach diverges from previous research on thin shells and plates, which primarily concentrated on deriving elasticity theories for material surfaces from classical three-dimensional Green elasticity. This work is the first in rigorously developing a surface elasticity model based on a parent nonlinear gradient elasticity theory.
\end{abstract}

\maketitle

\section{Introduction}

\subsection{Classical thin plates and shells}
Much of modern theoretical research into thin plates and shells has focused on either rigorously deriving surface elasticity theories from, or linking them to, classical three-dimensional Green elasticity. The idea (or expectation) is that passing to a two-dimensional surface model significantly decreases the complexity needed to solve problems for the three-dimensional body, while only slightly sacrificing accuracy. The mathematical methods involved include:
\begin{itemize}
\item using gamma convergence to obtain limiting variational problems  \cite{DretRaoult95, DreRaolt96, FriJamesMull02, FriJamesMu06}, 
\item performing asymptotic expansions of the weak and strong forms of the equilibrium equations \cite{Fox93, Ciarlet00Book, SongDai17, SongWangDai19},
\item obtaining leading order-in-thickness expressions for the kinetic and stored energy of a surface contained in the body from classical three-dimensional stored and kinetic energies \cite{HilgPip92b, HilgPip96, HilgPip97, Steig13, SteigShir19}.
\end{itemize}
One attractive aspect of the third approach, adopted in this paper, is its simplicity in incorporating both bending and stretching effects for the surface model. The equations governing the motion of the surface are subsequently derived by applying Hamilton's variational principle (see, e.g., \cite{HilgPip92a, Hilgers97}). 

An especially important example of a stored surface energy is \emph{Koiter's plate energy}. Consider a homogeneous, isotropic, nonlinearly elastic plate $\cl B = \cl S \times [-h/2, h/2] \subset \bbE^3$ with middle surface (midsurface) $\cl S$, thickness $h$, and stored energy per unit reference area $W$.\footnote{In this work, we denote three-dimensional Euclidean space by $\bbE^3$ and identify its translation space with $\bbR^3$ via a fixed orthonormal basis $\{ \bs e_i \}_{i = 1}^3$. Upon choosing an origin $\bs o \in \bbE^3$, we identify subsets of $\bbE^3$ with subsets of $\bbR^3$ via their position vectors: $\bbE^3 \ni \bs p \mapsto \bs p - \bs o \in \bbR^3$.  Throughout this work, we use standard vector and tensor operations in $\bbR^3$. We also raise and lower indices using the flat metric on $\bbR^3$, and we use the Einstein summation convention that repeated indices in upper and lower positions imply summation. 
	
Finally, we use standard big-oh and little-oh notation, e.g., $A = O(B)$ means that there exists $C \geq 0$ such that $|A| \leq C B$. We say that a big-oh term depends on $D$ if $C$ depends on $D$, $C = \hat C(D)$. } Here, $\cl S$ is a domain in the two-dimensional Euclidean plane. In what follows, Greek indices range in $\{1,2\}$. For a smooth motion of the midsurface, $\bs y : \cl S \times [t_0, t_1] \rar \bbE^3$, two key tensors on $\cl S$ are used to express Koiter's plate energy: the surface Green-Saint-Venant tensor $\msE$ and the relative curvature tensor $\msK$ defined via
\begin{gather}
	\msE = \sE_{\al \beta} \bs e^\al \tens \bs e^\beta, \quad \sE_{\al \beta} = \frac{1}{2}(\p_\al \bs y \cdot \p_\beta \bs y - \delta_{\al \beta}), \\
	\msK = \sK_{\al \beta} \bs e^\al \tens \bs e^\beta, \quad \sK_{\al \beta} = \bs n \cdot \p_{\al \beta}^2 \bs y, \quad \bs n = \frac{\p_1 \bs y \times \p_2 \bs y}{|\p_1 \bs y \times \p_2 \bs y|}.
\end{gather}
Koiter's plate energy $U_{Koiter}$ \cite{Koit66} is then given by
\begin{align}
 U_{Koiter} = 	h \Bigl ( \frac{\lambda \mu}{\lambda + 2\mu} (\tr \msE)^2 + \mu |\msE|^2 \Bigr )
 + \frac{h^3}{24} \Bigl ( \frac{\lambda \mu}{\lambda + 2\mu} (\tr \msK)^2 + \mu |\msK|^2 \Bigr ). \label{eq:Koiter}
\end{align}
In the above expression, $\lambda$ and $\mu$ represent the usual Lamé parameters of the material, $\tr \msE = \sE_\al^\al$, and $|\msE|^2 = \sE_{\al \beta} \sE^{\al \beta}$. One has a generalization of \eqref{eq:Koiter} for thin shells with curved midsurfaces, but for simplicity, we will restrict our discussion to plates. Steigmann \cite{Steig13} elegantly derived Koiter's plate (and shell) energy by expanding and integrating-in-thickness the three-dimensional stored energy under the assumption of small midplane strains. In particular, he showed that 
\begin{align}
	\int_{-h/2}^{h/2} W \, dZ = U_{Koiter} + o(h^3), \label{eq:Koiterderiv}
\end{align}
where $\bs y = \bs \chi |_{Z = 0}$ is a motion of the midsurface induced by a motion of the three-dimensional plate $\bs \chi : \cl B \times [t_0, t_1] \rar \bbE^3$. An earlier proof of \eqref{eq:Koiterderiv} was obtained by Hilgers and Pipkin in \cite{HilgPip96}, but there, the nature of the surface in $\cl B$ that is evolving according to $\bs y$ is unclear.\footnote{This is due to the fact that $\bs y := h^{-1}\int_{-h/2}^{h/2} \bs \chi \, dZ$ in their work.} In Chapter 4 of \cite{Ciarlet05Book}, Ciarlet presents a compelling body of evidence supporting the claim that \eqref{eq:Koiter} represents the ``best" stored surface energy for thin plates (and shells). His key argument is the fact that solutions to the linearized equilibrium equations\footnote{The equilibrium equations with prescribed tractions $\bs t$ applied on $\p \cl S$ correspond to the Euler-Lagrange equations for $V = \int_{\cl S} U_{Koiter} - \int_{\p \cl S} \bs t \cdot \bs u \, dA$ where $\msE$ and $\msK$ are replaced by their linearizations in the displacement field $\bs u$.} arising from \eqref{eq:Koiter} exhibit the same asymptotic behavior as $h \rar 0$, in the same function spaces, as displacements, averaged across thickness, that solve the full three-dimensional linearized elasticity problem. To the best of our knowledge, the only derivation of an analogous cubic-order-in-thickness kinetic energy for a plate is due to Hilgers and Pipkin \cite{HilgPip97}. They obtained the following expression,
\begin{align}
	K_{HP} = \frac{1}{2} h\rho_R \Bigl (
	|\p_t \bs y|^2 + \frac{h^2}{12} \frac{\la^2}{(\la + 2\mu)^2} |\tr \p_t \msE|^2 + \frac{h^2}{12} |\p_t \bs n|^2 
	\Bigr ), \label{eq:KHP}
\end{align} 
where $\rho_R$ is the plate's reference density. Again, the nature of the surface evolving by $\bs y$ is still ambiguous. Moreover, when combined with \eqref{eq:Koiter} to obtain a dynamic theory, the phase velocities for longitudinal infinitesimal harmonic plane waves have the degenerate property of vanishing in the short wavelength limit (see Section 4.3). 

As far as the author is aware, no previous effort has been made, using any method, to rigorously derive surface elasticity theories for thin plates or shells from three-dimensional, nonlinear, non-classical \emph{gradient} elasticity. This work represents the first in this new direction.    

\subsection{Gradient elasticity}

We now give a brief overview of (second-) gradient elasticity and our motivation for this work. In what follows, $\bs x = \bs \chi : \cl B \times [t_0, t_1] \rar \bbE^3$ is a smooth motion of a body with reference configuration $\cl B \subseteq \bbE^3$, reference density $\rho_R : \cl B \rar (0,\infty)$, and Eulerian velocity field $\bs v(\bs x,t)$.  We recall that classical Green elasticity posits that the stored energy $V$ and kinetic energy $T$ of a part $\cl P \subseteq \cl B$ at time $t$ are of the form, 
\begin{align}
		V(\cl P) = \int_{\cl P} W(\bs F)dV, \quad T(\cl P) = \int_{\cl P} \frac{1}{2} \rho_R |\p_t \bs \chi|^2 dV, 
\end{align}    
where $\bs F = \p_{X^a} \chi^i \bs e_i \tens \bs e^a$ is the deformation gradient. The equations of motion neglecting external body forces and boundary tractions can then be obtained via applying Hamilton's variational principle to the action functional with Lagrangian density $L = T - V$.  

\emph{Gradient elasticity} is a subtheory of gradient continuum mechanics in which the kinetic energy and stored energy densities can also depend on the spatial derivatives of $\p_t \bs \chi$ and $\bs F$ respectively \cite{Maugin17Book}. As before, the equations of motion are then obtained via Hamilton's variational principle. Piola was the first to conceive of continua where the body's internal work expenditures depend on spatial derivatives of higher order than classical Cauchy continua, dating back to 1846 \cite{Piola1846Book, dellIsola15}. However, significant progress in this area didn't emerge until the latter half of the 20th century. During this period, a surge of activity by prominent figures including Toupin \cite{Toupin62, Toupin64}, Green and Rivlin \cite{GreenRivlin64a, GreenRivlin64b}, Mindlin \cite{Mindlin64a, Mindlin1965}, Mindlin and Eshel \cite{MindlinEshel1968}, and Germain \cite{Germain73a, Germain73b} resulted in the development of comprehensive theories of gradient continua, including gradient elasticity. A detailed review of the extensive work and applications of gradient continuum models since then is beyond the scope of this work. For further insights, we refer to the reviews \cite{Askes2011, dellIsola17, Maugin17Book, dellIsola2020higher} and the references cited therein.

In theories of gradient continua, an added layer of complexity emerges since it becomes necessary to prescribe additional boundary conditions. In particular, for the equilibrium theory derived from a stored energy density that depends on the spatial gradient of $\bs F$ and satisfies strong ellipticity, the resulting equilibrium equations form an elliptic system of partial differential equations that include the fourth-order spatial derivatives of the configuration $\bs \chi$. It is then necessary to specify a boundary condition in addition to the traditional prescription of either boundary placement, $\bs \chi |_{\p \cl B}$, or boundary surface tractions, $\bs P \bs N |_{\p \cl B}$, to ensure a well-posed boundary value problem in general.\footnote{Here, $\bs N$ is the outward-pointing normal vector field on $\p \cl B$ and $\bs P$ is the Piola-Kirchhoff tensor, which, in the classical setting is given by $\bs P = \frac{\p W}{\p F^i_a} \bs e^i \tens \bs e_a$. For strain energy densities depending also on the spatial gradient of $\bs F$, the form of $\bs P$ is different (see \eqref{eq:Pdefinition}).} For the case of placement, it is mathematically natural (but perhaps unclear from a physical standpoint) to additionally specify the normal derivative of the configuration $\p_N \bs \chi |_{\p \cl B}$. For the case of boundary surface tractions, it is physically natural (but perhaps mathematically cumbersome) to also specify the distribution of boundary surface couples since these arise (in addition to boundary surface tractions) as \textit{natural boundary conditions} from a variational point of view (see, e.g., Sections 3-4 in \cite{Mindlin64a}, Sections 5-6 of \cite{Toupin64}, Section 5 of \cite{Germain73a}). The analogous problem of determining boundary conditions for Rivlin-Ericksen fluids is briefly discussed in Section 6.2 of \cite{TruesdellRajBook00} and thoroughly discussed in Section 2.5 of \cite{CGRajBook16}. We emphasize that their overarching philosophy applies equally well in the case of solids. In particular, when prescribing boundary conditions for a specific problem, convenience should not be the criterion. Instead, these conditions ``should reflect some physical idea" \cite{TruesdellRajBook00} related to the situation at hand and involve considering both ``the structure of the material that is being enveloped by the boundary as well as the structure of the environment" \cite{CGRajBook16}.  

In this work, we consider especially simple energies for a homogeneous, isotropic, gradient elastic body that introduce the smallest number of additional constants: the frame-indifferent stored energy is 
\begin{align}
	V(\cl P) = \int_{\cl P} \Bigl [ \frac{\la}{2} (\tr \bs E)^2 + \mu |\bs E|^2 + \ell_s^2 \sum_{c = 1}^3 \Bigl ( \frac{\la}{2} (\p_{X^c} \tr \bs E)^2 + \mu |\p_{X^c} \bs E|^2 \Bigr ) \Bigr ] dV, \label{eq:poten}
\end{align} 
and the kinetic energy is
\begin{align}
	T(\cl P) &= \int_{\cl P} \frac{1}{2} \rho_R \Bigl ( |\p_t \bs \chi|^2 + \ell_k^2 |\p_t \bs F (\bs F)^{-1}|^2 \Bigr ) dV  \label{eq:kinen} \\ &= \int_{\bs \chi(\cl P)} \frac{1}{2} \rho \Bigl ( |\bs v|^2 + \ell_k^2 |\mbox{grad}\, \bs v|^2 \Bigr ) dv,
\end{align} 
where $\mbox{grad}\, \bs v = \p_{x^j} v^i \bs e_i \tens \bs e^j$ is the velocity gradient, $\rho = \rho_R (\det \bs F)^{-1}$ is the current density, and $\ell_s$ and $\ell_k$ are two additional positive length parameters. For the case of infinitesimal motions, the kinetic energy \eqref{eq:kinen} was first suggested by Mindlin \cite{Mindlin64a}, and the stored energy \eqref{eq:poten} appeared first in Aifantis and Altan's work \cite{AltanAif1992}, inspired by earlier work of Aifantis \cite{Aifantis1984, Aifantis1987} and Aifantis and Triantafyllidis \cite{TriantafyllidisAifantis1986} on gradient continua.  The values of the length parameters $\ell_s$ and $\ell_k$ in terms of conceivably measurable physical properties has been a source of debate for some years with most values given in terms of the body's natural inter-particle spacing $d$ (see, e.g., Section 4 of \cite{Askes2011}). 


In a recent work by the author \cite{Rodriguez2023StrainGradient}, an equilibrium theory was developed for a classical three-dimensional Green elastic bulk solid with a gradient elastic boundary surface. This theory was applied to a mode-III fracture problem, effectively eliminating the problematic singularities present in both stresses and strains that arise from classical linear elastic fracture mechanics. The model's stored surface energy 
\begin{align}
	U_{HP} = 	h \Bigl ( \frac{\lambda \mu}{\lambda + 2\mu} (\tr \msE)^2 + \mu |\msE|^2 \Bigr )
	+ \frac{h^3}{24} \Bigl ( \frac{\lambda \mu}{\lambda + 2\mu} \bigl | \p_\al \p^\al \bs y
	\bigr |^2 + \mu \sum_{\al, \beta} |\p_\al \p_\beta \bs y|^2  \Bigr ) \label{eq:HP}
\end{align}
was suggested in Hilgers and Pipkin's research \cite{HilgPip92a, HilgPip96} as an ad hoc modification of $U_{Koiter}$ that satisfies the strong ellipticity condition. The fulfillment of the strong ellipticity condition by $U_{HP}$ played a crucial role in demonstrating that the ensuing fracture model generated bounded stresses and strains up to the crack tips. However, the connection of $U_{HP}$ to a parent three-dimensional theory, in the spirit of work discussed in the previous section, is unclear. The primary motivation and outcome of this work involve the introduction of an attractive alternative quadratic stored surface energy, denoted as $U$, such that:
\begin{itemize}
\item $U$ satisfies the strong ellipticity condition and can be used in conjunction with \cite{Rodriguez2023StrainGradient} to eliminate the singularities present in linear elastic fracture mechanics. 
\item $U$ is also expressed in terms of (conceivably measurable) physical properties of the plate.
\item $U$ is derived from a parent three-dimensional theory, in the spirit of \eqref{eq:Koiterderiv}.
\end{itemize}
\subsection{Main results and outline}

For simplicity, in this work we will only consider the case of plates $\cl B = \cl S \times [-h/2,h/2]$, but we expect our results can be generalized to shells with curved midsurfaces by using more differential geometric machinery (see, e.g., \cite{Ciarlet05Book, Steig13}). In Section 2, we first discuss the necessary kinematics and set-up for our study. We then formally argue via averaging linearized lattice dynamics that for a plate with kinetic and stored energies \eqref{eq:kinen} and \eqref{eq:poten}, reasonable identifications of the length scales are $\ell_s^2 = d^2/12$ and $\ell_k^2 = d^2/6,$ where $d$ is the inter-particle spacing of the physical plate in its natural configuration. 

In Section 3, we perform an asymptotic expansion-in-thickness to obtain cubic order expressions for the stored surface energy and kinetic energy of the plate's midsurface in the spirit of \eqref{eq:Koiterderiv}.  Based on the identifications indicated in Section 2, we assume that $\ell_s^2 + \ell_k^2 \leq C_0 h^2$ where $C_0$ is a fixed constant. This mathematical restriction reflects the reasonable physical assumption that the plate's thickness is not drastically smaller than the inter-particle spacing. We assume the motion $\bs \chi$ has small strain and normal components of the second gradient of the velocity (relative to fixed length and time scales) when evaluated on the midplane. We then prove that the associated stored surface energy per unit reference area, $\int_{-h/2}^{h/2} W \, dZ$, and midsurface kinetic energy per unit reference area, $\int_{-h/2}^{h/2} \kappa_R \, dZ$, where $W$ is the integrand appearing in \eqref{eq:poten} and $\kappa_R$ is the integrand appearing in \eqref{eq:kinen}, satisfy
\begin{gather}
	\int_{-h/2}^{h/2} W \, dZ = 
	h \Bigl [ \frac{\lambda \mu}{\lambda + 2\mu} (\tr \msE)^2 + \mu |\msE|^2  
	+ \ell_s^2 \sum_{\gamma = 1}^2 \Bigl ( \frac{\lambda \mu}{\lambda + 2\mu} (\tr \p_\gamma \msE)^2 + \mu |\p_\gamma \msE|^2 \Bigr ) \Bigr ] \\
	+ 
	\Bigl (\frac{h^3}{24} + h \ell_s^2 \Bigr ) \Bigl ( \frac{\lambda \mu}{\lambda + 2\mu} (\tr \msK)^2 + \mu |\msK|^2 \Bigr ) + O(h^4), \label{eq:stoexp}
\end{gather}
and 
\begin{gather}
	\int_{-h/2}^{h/2} \kappa_R \, dZ = \frac{1}{2} h \rho_R \Bigl (
	|\p_t \bs y|^2 + \frac{h^2+12\ell_k^2}{12} \frac{\la^2}{(\la + 2\mu)^2} |\tr \p_t \msE|^2 + \frac{h^2+12 \ell_k^2}{12} |\p_t \bs n|^2 \\ + \ell_k^2 |\nabla \p_t \bs y|^2 
	\Bigr ) + O(h^{7/2}). \label{eq:kinexp}
\end{gather}
Here, $\bs y = \bs \chi |_{Z = 0}$ and the gradient operator is with respect to the midsurface variables: see Theorem \ref{t:t1} for the precise statement.  We denote by $U$ and $K$ the leading order expressions appearing in \eqref{eq:stoexp} and \eqref{eq:kinexp} respectively, and we observe that they are characterized by five physical properties of the plate by using our identifications for the length scales: its reference density $\rho_R$, Young's modulus $E = \frac{\mu(2\mu+3\lambda)}{\mu + \lambda}$, Poisson's ratio $\nu = \frac{\lambda}{2(\mu + \lambda)}$, thickness $h$, and natural configuration's inter-particle spacing $d$ (see \eqref{eq:surfacestored} and \eqref{eq:surfacekinetic}. In the singular limiting case $\ell_s = 0$, we recover Koiter's classical shell energy \eqref{eq:Koiter}, and in the case $\ell_k = 0$ we recover Hilgers and Pipkin's kinetic energy \eqref{eq:KHP}. We conclude Section 3 by showing that $U$ satisfies the strong ellipticity condition (rather than the weaker Legendre-Hadamard condition) precisely when $\ell_s > 0$ (see Proposition \ref{p:p1}). 

In Section 4, we adopt $U$ and $K$ as the stored surface energy and kinetic energy densities of an elastic material surface, respectively. By applying Hamilton's variational principle, we establish the field equations governing the motion of the surface. Our analysis then focuses on infinitesimal harmonic plane waves, providing a detailed comparison between the setting when $\ell_s$ and $\ell_k$ are non-zero versus the case where $\ell_s = \ell_k = 0$. Specifically, waves moving in all directions (longitudinal, tangentially transverse, and normally transverse) exhibit dispersion. Unlike the scenario with $\ell_s = \ell_k = 0$, the phase velocity of longitudinal waves is now consistently bounded below by a fixed positive constant for all large wave numbers. 

\begin{figure}[t]
	\centering
	\includegraphics[scale=.4]{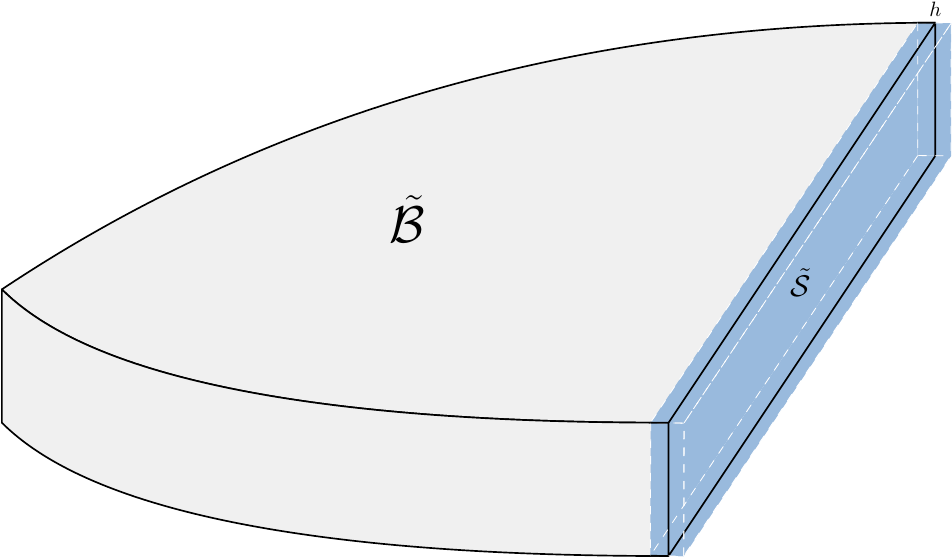}
	\caption{A body primarily made of a Green elastic material with an additional thin, gradient elastic region of thickness $h$ extending from a section of its boundary and our modeling scheme.}
	\label{f:1}
\end{figure}

In the final section, we show how utilizing the stored energy density $U$ in the framework proposed by the author in \cite{Rodriguez2023StrainGradient} eliminates the singularities in stresses and strains observed in linear elastic fracture mechanics, particularly in the context of mode-III fracture. The overall model can be physically understood as representing a body primarily made of a Green elastic material with an additional thin, gradient elastic region of thickness $h$, extending from a section of its boundary. More precisely, we model such a body by a Green elastic solid $\tilde{\cl B}$ possessing a gradient elastic boundary surface, $\tilde{\cl S}$, with stored surface energy density $U$; see Figure \ref{f:1} for a schematic. In particular, the fracture model is parameterized by the three-dimensional solid's Young's modulus, Poisson's ratio, and natural inter-particle spacing, and the sole parameter left undetermined is the thickness of the region emanating from the crack front where gradient elastic, small-scale processes become significant.

\subsection*{Acknowledgments} The author was supported by NSF Grant DMS-2307562.

\section{Preliminaries}

\subsection{Kinematics}
Let $\cl B = \cl S \times  [ -h/2, h/2 ] \subseteq \bbE^3$ where $\cl S$ is a smooth domain in $\bbR^2$. The domain $\cl B$ is the reference configuration of a three-dimensional plate, and $\cl S$ is its two-dimensional \emph{midsurface}. For a reference element with position vector $\bs X = X^a \bs e_a \in \cl B$, we write $\bs X = \bs Y + Z \bs e_3 = Y^\al \bs e_\al + Z \bs e_3$ where Greek indices range over $\{ 1, 2 \}$.   

For a smooth motion $\bs x = \hat{\bs \chi} = \hat \chi^i \bs e_i : \cl B \times [t_0, t_1] \rar \bbE^3$, we denote the deformation gradient and second deformation gradient by
\begin{gather}
	\hat{\bs F} = \p_a \hat \chi^i \bs e_i \tens \bs e^a, \quad \hat{\bs A} = \p^2_{ab} \hat \chi^i \bs e_i \tens \bs e^a \tens \bs e^b, 
\end{gather} 
where $\bs e^a = \bs e_a$ for $a = 1,2$, and $\p_a = \frac{\p}{\p X^a}$. The velocity field in Eulerian coordinates $\bs x$ is denoted by $\hbv(\bs x, t) = \p_t \hbchi (\hbchi^{-1}(\bs x), t)$. The Green-Saint-Venant strain tensor is given by $\hbE = \frac{1}{2}(\hbF^T \hbF - \bs I)$. We denote the normal to the convected midsurface $\hbchi(\cl S)$ by 
\begin{align}
	\bs n = |\p_1 \hbchi \times \p_2 \hbchi|^{-1} \p_1 \hbchi \times \p_2 \hbchi \big |_{Z = 0}, 
\end{align} and the surface Green-Saint-Venant strain tensor $\msE$ and relative curvature tensor $\msK$ for the midsurface are defined by 
\begin{gather}
	\msE = \p_{\al} \hbchi \cdot \p_\beta \hbchi \big |_{Z = 0} \bs e^\al \tens \bs e^\beta = \hat E_{\al \beta}\big |_{Z = 0} \bs e^\al \tens \bs e^\beta, \\
	\msK = \bs n \cdot \p_{\al \beta}^2 \hbchi \big |_{Z = 0} \bs e^\al \tens \bs e^\beta.  
\end{gather}
In what follows, the absence of a caret indicates that the quantity is evaluated at $Z = 0$, e.g., $E_{\al \beta} = \hat{E}_{\al \beta} |_{Z = 0}$, and indices are raised and lowered using the flat metric on $\bbR^3$.   

\subsection{Stored energy and kinetic energy assumptions}

We assume that the plate $\cl B$ is homogeneous, isotropic, the reference configuration is its natural configuration, and the stored energy per unit reference volume is of the simple strain-gradient form,  
\begin{gather}
	W = \frac{\la}{2} (\tr \hbE)^2 + \mu |\hbE|^2 + \ell_s^2 \sum_{c = 1}^3 \Bigl ( \frac{\la}{2} (\p_c \tr \hbE)^2 + \mu |\p_c \hbE|^2 \Bigr ), \label{eq:stored}
\end{gather}
where $\la$ and $\mu$ are the usual Lam\'e parameters of the material and $\ell_s$ is an internal length parameter. We denote
\begin{align}
	W_s = \frac{\la}{2} (\tr \hbE)^2 + \mu |\hbE|^2, \quad W_{sg} = \sum_{c = 1}^3 \Bigl ( \frac{\la}{2} (\tr \p_c \hbE)^2 + \mu |\p_c \hbE|^2 \Bigr ).
\end{align}
We will also assume that the kinetic energy per unit reference volume $\kappa_R$ of the plate includes a simple velocity gradient contribution,  
\begin{align}
	\kappa_R = \frac{1}{2} \rho_R \Bigl ( |\p_t \hbchi|^2 + \ell_k^2 |\p_t \hbF (\hbF)^{-1}|^2 \Bigr ), \label{eq:kin}
\end{align}
where $\rho_R$ is the constant, reference mass density and $\ell_k$ is a second internal length parameter. We note that in terms of the velocity field $\hbv$, the kinetic energy per unit current volume, $\kappa$, takes a more conventional looking form    
\begin{gather}
		\kappa = \frac{1}{2} \rho \Bigl ( |\bs \hbv|^2 + \ell_k^2 |\mbox{grad}\, \hbv|^2 \Bigr ), \label{eq:kincurr} \\
	\rho = \rho_R (\det \hbF)^{-1}, \quad \mbox{grad}\, \hbv = \frac{\p \hat{v}^i}{\p x^j} \bs e_i \tens \bs e^j.
\end{gather}
As discussed in the introduction, the stored energy \eqref{eq:stored}, kinetic energy \eqref{eq:kin} (equivalently \eqref{eq:kincurr}), and the identification of the additional length parameters $\ell_s$ $\ell_k$ have been intensely studied for several decades.

We define the following second-order Piola-Kirchhoff tensors $\hbP$ and $\hbP_s$ corresponding to the stored energies $W$ and $W_s$ by 
\begin{gather}
	\hbP = \Bigl (\frac{\p W}{\p \hat{F}^i_a} - \frac{\p}{\p X^b} \frac{\p W}{\p \hat{A}^i_{ab} } + \frac{\p}{\p t} \frac{\p \kappa_R}{\p \dot{\hat{F}}^i_{a}} \Bigr ) \bs e^i \tens \bs e_a, \quad 
		\hbP_s = \frac{\p W_s}{\p \hat{F}^i_a} \bs e^i \tens \bs e_a, \label{eq:Pdefinition}
\end{gather}
where $\dot{\hat{F}}^i_{a} = \p_t \hat{F}^i_{a}$. As is well-known (see e.g., \cite{TruesdellNollNLFT, Toupin64}), for a body with stored energy $W$ (or $W_s$), the resultant contact force  exerted on a subdomain $\cl R \subseteq \cl B$ by $\cl B \backslash \cl R$ is given by 
$\bs F(\cl R) = \int_{\cl R} \hbP \bs N \, dA$, where $\bs N$ is the unit normal vector field on $\p \cl R$ pointing into $\cl B \backslash \cl R$.{\footnote{The expression for $\hbP_s$ is classical while the static expression for $\hbP$ can be found in Section 10 of \cite{Toupin64}. The dynamic expression for $\hbP$ can be deduced via similar integration by parts calculations as in \cite{Toupin64}.}  From our expressions for $W$ and $\kappa_R$, we have 
\begin{align}
	\hbP = \hbP_s + O(\ell_s^2 + \ell_k^2),\label{eq:PPscomparison}
\end{align}
where the big-oh term depends on the size of $\| \hbchi \|_{C^3(\cl B \times [t_0,t_1])}$.

\subsection{Identification of the length parameters via averaging lattice dynamics}
We now give a \emph{formal} argument for averaging lattice dynamics that suggests the identifications 
\begin{align}
	\ell_s = \frac{1}{\sqrt{12}} d, \quad \ell_k = \frac{1}{\sqrt{6}} d, \label{eq:specificls}
\end{align}
where $d > 0$ is the inter-particle spacing for the plate's natural, unstressed configuration.

We model the plate of thickness $h$ by a three-dimensional lattice $\cl L_d$ of identical particles with spacing $d > 0$, each with mass $M$, and with reference positions 
\begin{align}
	\bs X_{d,\bs m} = (m_1 d, m_2d, m_3 d), \quad \bs m = (m_1, m_2, m_3) \in \bbZ^3. 
\end{align}
Let $\bs u_{d,\bs m}(t)$ be the displacement of the particle with reference position $\bs X_{d,\bs m}$ at time $t$, and we assume that the the resultant force $\bs F_{d,\bs m}$ exerted on the particle with reference position $\bs X_{d,\bs m}$ is a superposition of forces exerted by its nearest neighbors of linear Hookean type, 
\begin{align}
	\bs F_{d,\bs m} = k_d \sum_{j = 1}^3 \sum_{\pm} (\bs u_{d,\bs m \pm d\bs e_j} - \bs u_{d, \bs m}), \label{eq:resultant}
\end{align} 
where $k_d > 0$. The dynamics of the lattice are then determined by Newton's laws of motion 
\begin{align}
	M \ddot{\bs u}_{d,\bs m}(t) = k_d \sum_{j = 1}^3 \sum_{\pm} (\bs u_{d,\bs m \pm d\bs e_j}(t) - \bs u_{d, \bs m}(t)), \quad \bs m \in \bbZ^3. \label{eq:newtonlattice}
\end{align}

We now assume that there exists a vector field $\bs u(\cdot, t) : \cl B \rar \bbR^3$ such that for each time $t$, and $\bs m \in \bbZ^3$, 
\begin{align}
	\bs u_{d,\bs m}(t) = \bs u(\bs X_{d,\bs m}, t) + o(d^4), 
\end{align} 
We interpret $\bs u(\bs X,t)$ as the displacement at time $t$ of the particle with reference position $\bs X$ in the homogenized plate modeled by $\cl B$. For $d > 0$, the \emph{average displacement} of the $d$-cell centered at $\bs X$ is defined by  
\begin{align}
	\msu(\bs X,t;d) = \frac{1}{8 d^3} \int_{X_3 - d}^{X_3 + d} \int_{X_2 - d}^{X_2 + d} \int_{X_1 - d}^{X_1 + d} \bs u(\tilde{\bs X},t) \, dv. 
\end{align}
For a shearing motion with $\bs u_{d,\bs m} = u_{d,m_1} \bs e_2$ and $\bs u(\bs X,t) = u(X^1,t) \bs e_2$, we have that for all $\bs m \in \bbZ^3$,  
\begin{align}
	\bs F_{d,\bs m} = k_d [u(m_1d + d, t) + u(m_1d - d,t) - 2u(m_1d, t) + o(d^4)] \bs e_2 . \label{eq:resultantshear}
\end{align}  

In what follows, we denote $X^1$ by $X$, $m_1$ by $m$ and $X_{d,m} = md$. By Newton's laws, \eqref{eq:newtonlattice}, and \eqref{eq:resultantshear}, the motion satisfies for all $m \in \bbZ, t \in \bbR$, 
\begin{gather}
	M [\p_t^2 u(X_{d,m},t)+o(d^4)] \\ = k[u(X_{d,m+1},t) + u(X_{d,m-1},t) - 2 u(X_{d,m},t) + o(d^4)]. \label{eq:newton}
\end{gather}
The average displacement $\msu = \su(X,t;d) \bs e_2$ with $\su(X,t;d) = \frac{1}{2d} \int_{X-d}^{X+d} u(\tilde X,t) d\tilde X$, and by Taylor's theorem, for all $m \in \bbZ$, 
\begin{gather}
	\p_t^2 \su(X_{d,m},t;d) = \p_t^2 u(X_{d,m},t) + \frac{d^2}{6} \p_t^2 \p_X^2 u(X_{d,m},t) + O(d^4), \\
	\p_X^2 \su(X_{d,m},t;d) = \p_X^2 u(X_{d,m},t) + \frac{d^2}{6} \p_X^4 u(X_{d,m},t) + O(d^4), \\
	\p_X^4 \su(X_{d,m},t;d) = \p_X^4 u(X_{d,m},t) + O(d^2), \\
	u(X_{d,m+1},t) + u(X_{d,m-1},t) - 2 u(X_{d,m},t) = d^2 \Bigl ( \p_X^2 u(X_{d,m},t) \\ + \frac{1}{12}d^2 \p_X^4 u(X_{d,m},t) + O(d^4)\Bigr ).
\end{gather}
Thus, by \eqref{eq:newton} and for all $X \in \{ md \mid m \in \bbZ\}$,
\begin{align}
	\rho_d \Bigl (
	\p_t^2 \su(X,t;d) - \frac{d^2}{6} &\p_X^2 \p_t^2 \su(X,t;d) + o(d^2)
	\Bigr ) \\&= G_d \Bigl ( \p_X^2 \su(X,t;d) - \frac{d^2}{12} \p_X^4 \su(X,t;d) + o(d^2) \Bigr ). \label{eq:homogeq}
\end{align}
where $\rho_d = M/dh^2$ is the mass density per unit reference volume and $G_d = k_dd/h^2$ is the shear modulus.  

We now assume that the plate modeled via $\cl B$ has stored energy and kinetic energy densities $W$ and $\kappa_R$. The nonlinear field equations governing motion of the body are then derived via Hamilton's variational principle applied to the action $\cl A([t_0,t_1]) = \int_{t_0}^{t_1} \int_{\cl B} [\kappa_R - W] dV dt$. One may then verify that the linearized equations governing infinitesimal shearing of $\cl B$,
$
	\hbchi = \bs X + w(X^1,t) \bs e_2,
$
are given by
\begin{align}
		\p_t^2 w(X,t) - \ell_k^2 \p_X^2 \p_t^2 w(X,t) = c^2_T 
		\Bigl ( \p_X^2 w(X,t) - \ell_s^2 \p_X^4 w(X,t) \Bigr ) \label{eq:longeq}
\end{align} 
where $c^2_T = \mu/\rho_R$ is the shear speed. Assuming that the continuum model consistently approximates the average lattice dynamics in the sense that $G_d/\rho_d \rar c^2_T$ and for $d > 0$, $w = \su$ solves \eqref{eq:longeq} up to an $o(d^2)$ error, then \eqref{eq:homogeq} and \eqref{eq:longeq} imply the identifications 
\begin{align}
	\ell_s^2 = \frac{d^2}{12}, \quad \ell_k^2 = \frac{d^2}{6}. 
\end{align} 

\section{Surface energies obtained as the leading order expressions from the plate energies}

In this section, we obtain leading cubic order-in-$h$ expressions for the stored surface and kinetic energy densities 
\begin{align}
	\int_{-h/2}^{h/2} W \, dZ \quad \mbox{and} \quad \int_{-h/2}^{h/2} \kappa_R \, dZ, \label{eq:storedkinetic}
\end{align}
corresponding to motions with Taylor series
\begin{align}
\hbchi(\bs X, t) = \bs y(\bs Y, t) + \bs d(\bs Y, t) Z + \bs g(\bs Y, t) \frac{Z^2}{2} + \bs h(\bs Y, t) \frac{Z^3}{6} + O(Z^4). \label{eq:chiexp}
\end{align}
Here, $\bs y = \hat{\bs \chi} |_{Z = 0}$ is the motion of the plate's midsurface $\cl S$, and the leading order expressions obtained involve kinematic quantities associated to $\bs y$ (see Theorem \ref{t:t1}).

The deformation gradient and second deformation gradient satisfy
\begin{align}
\hbF &= \grad \bs y + \bs d \tens \bs e_3 + (\nabla \bs d + \bs g \tens \bs e_3) Z 
+ (\nabla \bs g + \bs h \tens \bs e_3)\frac{Z^2}{2} + O(Z^3), \\
\hbA &= \nabla \nabla \bs y + (\nabla \bs d \tens \bs e_3)^T + \nabla \bs d \tens \bs e_3 + 
\bs g \tens \bs e_3 \tens \bs e_3 + O(Z),
\end{align}
where $\nabla$ is the gradient operator with respect to $\bs Y$,
\begin{align}
\forall \bs a = a^i \bs e_i, \quad \nabla \bs a = \p_{\al} a^i \bs e_i \tens \bs e^\al, \quad \nabla \nabla \bs a = \p_{\beta} \p_{\al} a^i \bs e_i \tens \bs e^\al \tens \bs e^\beta,
\end{align}
and $(\bs a \tens \bs b \tens \bs c)^T = \bs a \tens \bs c \tens \bs b$. Throughout this paper, we will also often use the notation $\bs a_{,\al} = \p_\al \bs a$ where $\al = 1, 2$. For the following discussion, we define the following quantities,
\begin{gather}
	\bs F^{(n)} = \p_{Z}^n \hbF |_{Z = 0}, 
	\quad \bs P_s^{(n)} = \p_{Z}^n \hbP_s(\hbF) |_{Z = 0}, 
\end{gather}
so that 
\begin{gather}
	\bs F = \nabla \bs y + \bs d \tens \bs e_3, \quad \bs F' = \nabla \bs d + \bs g \tens \bs e_3, \quad 	\bs F'' = \nabla \bs g + \bs h \tens \bs e_3, \\
\bs A = \nabla \nabla \bs y + (\nabla \bs d \tens \bs e_3)^T + \nabla \bs d \tens \bs e_3 + 
\bs g \tens \bs e_3 \tens \bs e_3.
\end{gather}

\subsection{Size assumption between the length parameters and plate thickness} 

For the remainder of this work, we fix characteristic length and time scales $L$ and $T$, and we will assume that $h/L$ is much smaller than unity. In addition, we will later assume that the normal components of the second gradient of the velocity evaluated on the midsurface satisfies $O(h/L^2 T)$ (to be made more precise below, see Theorem \ref{t:t1}). We adopt $T$ as our unit of time and $L$ as our measure of length. After proper nondimensionlization and relabeling of the variables, we may assume that $\bs X$ and $t$ are dimensionless, $h$ is a dimensionless small parameter. Finally, we will assume that there exists a dimensionless constant $C_0 > 0$ such that
\begin{align}
	\ell_s^2 + \ell_k^2 \leq C_0 h^2. \label{eq:hellcomp}
\end{align}
The identifications made in the previous section (\eqref{eq:specificls}) suggest that \eqref{eq:hellcomp} reflects the reasonable physical assumption that the plate's thickness is not significantly less than the natural inter-particle spacing.

\subsection{Parameterizing the motions}

Following the approaches of Hilgers and Pipkin \cite{HilgPip92b, HilgPip96, HilgPip97} and Steigmann \cite{Steig13}, we require that the generalized tractions $$\hbP \bs e_3 |_{Z = \pm h/2}$$ vanish to first order in $h$. The resulting conditions imposed on $\bs d$ and $\bs g$ are as follows. By \eqref{eq:PPscomparison} and \eqref{eq:hellcomp},
\begin{align}
	\hbP \bs e_3 |_{Z = \pm h/2} 
	&= \bs P_s \bs e_3 \pm \bs P_s' \bs e_3 \frac{h}{2} + O(\ell_s^2 + \ell_k^2) + O(h^2) \\
	&= \bs P_s \bs e_3 \pm \bs P_s' \bs e_3 \frac{h}{2} + O(h^2),
\end{align}
where the final big-oh term depends on $\| \hbchi \|_{C^3(\cl B \times [t_0,t_1])}$.
Therefore, $\hbP \bs e_3 |_{Z = \pm h/2} = \bs 0$ to first order in $h$ if and only if $\bs P_s \bs e_3 = \bs 0$ and $\bs P_s' \bs e_3 = \bs 0$. 
We compute 
\begin{align}
	\bs P_s \bs e_3 = \frac{\p W_s}{\p E_{c3}} \p_c \bs \chi, \quad  
	\bs P_s' \bs e_3 = \frac{\p^2 W_s}{\p E_{ab} \p E_{c3}} \p_3 E_{ab} \p_c \bs \chi + \frac{\p W_s}{\p E_{c3}} \p^2_{c3} \bs \chi,  
\end{align}
and thus, $\bs P_s \bs e_3 = \bs 0$ and $\bs P_s' \bs e_3 = \bs 0$ if and only if for $c = 1, 2, 3$ we have
\begin{align}
	\frac{\p W_s}{\p E_{c3}} = 0, \quad \frac{\p^2 W_s}{\p E_{ab} \p E_{c3}} \p_3 E_{ab} = 0. \label{eq:U1relations}
\end{align}
The relations \eqref{eq:U1relations} are equivalent to
\begin{gather}
	E_{\al 3} = 0, \quad \al = 1,2, \quad E_{33} = -\frac{\la}{\la + 2\mu} \tr \msE, \label{eq:Erelations}\\
	\p_3 E_{\al 3} = 0, \quad \al = 1,2, \quad \p_3 E_{33} = -\frac{\la}{\la + 2\mu} \p_3 \tr \msE, \label{eq:p3Erelations}
\end{gather}
and thus, 
\begin{gather}
	\bs d = \phi \bs n, \quad
	\phi^2 = 1 + 2 E_{33} = 1 - \frac{2\lambda}{\lambda + 2\mu} \tr \msE, \label{eq:bard}
\end{gather}
The precise form of $\bs g$ will be unnecessary, but it will be useful to have it's structure on hand. Let $\tensor{M}{_i^k_j^l} = \frac{\p^2 W_s}{\p F^i_k \p F^j_l}$ and $A_{ij} = \tensor{M}{_i^3_j^3}$. We note that $\bs A = A_{ij} \bs e^i \tens \bs e^j$ is invertible for small $h$ since $W_s$ is convex in a neighborhood of $\bs E = \bs 0$. Then $\bs P_s' \bs e_3 = \bs 0$ is equivalent to 
\begin{align}
	g^i = -(A^{-1})^{ij} \tensor{M}{_j^3_k^\al} \p_\al d^k.  \label{eq:formofg}
\end{align}  
 
\subsection{Leading cubic order-in-$h$ expressions for the surface energies} 

Our main result of this section is the following cubic order-in-$h$ expansion of the stored surface and kinetic energies from \eqref{eq:storedkinetic}.   
\begin{thm}\label{t:t1}
Let $C_0, C_1$, and $C_2$ be fixed positive numbers, and assume that 
\begin{align}
	\ell_s^2 + \ell_k^2 \leq C_0 h^2. \label{eq:dhcomp}
\end{align}
Let $\hbchi \in C^3(\cl B \times [t_0, t_1])$ be a motion such that for each $t \in [t_0, t_1]$, $\hbchi(\cdot,t)$ is an immersion satisfying \eqref{eq:chiexp}, \eqref{eq:Erelations}, \eqref{eq:p3Erelations}. Assume that $\hbchi$ satisfies the a priori bounds
\begin{gather}
	\| \hbchi \|_{C^3(\cl B \times [t_0, t_1])} \leq C_1, \label{eq:chibound} \\
	\quad \forall \bs Y \in \bbR^2, t \in [t_0, t_1], \quad \bigl |\mathrm{Grad} \, \p_t {\bs F}(\bs Y,t)[\bs N \tens \bs N] \bigr | \leq C_2 h, \quad |\msE(\bs Y,t)| \leq C_2 h, \qquad \label{eq:strainconst2}   
\end{gather}
where $\bs N = \bs e_3$ is the (constant) normal vector field on the midplane.
  
Then 
\begin{gather}
	\int_{-h/2}^{h/2} W \, dZ = 
	h \Bigl [ \frac{\lambda \mu}{\lambda + 2\mu} (\tr \msE)^2 + \mu |\msE|^2  
	+ \ell_s^2 \sum_{\gamma = 1}^2 \Bigl ( \frac{\lambda \mu}{\lambda + 2\mu} (\tr \p_\gamma \msE)^2 + \mu |\p_\gamma \msE|^2 \Bigr ) \Bigr ] \\
	+ 
	\Bigl (\frac{h^3}{24} + h \ell_s^2 \Bigr ) \Bigl ( \frac{\lambda \mu}{\lambda + 2\mu} (\tr \msK)^2 + \mu |\msK|^2 \Bigr ) + O(h^4), \label{eq:storedexpansion}
\end{gather}
and 
\begin{gather}
	\int_{-h/2}^{h/2} \frac{1}{2} \rho_R \Bigl (
	|\p_t \hbchi|^2 + \ell_k^2 |\p_t \hbF (\hbF)^{-1}|^2 
	\Bigr ) dZ \\ = \frac{1}{2} \rho_s \Bigl (
	|\p_t \bs y|^2 + \frac{h^2+12\ell_k^2}{12} \frac{\la^2}{(\la + 2\mu)^2} |\tr \p_t \msE|^2 + \frac{h^2+12 \ell_k^2}{12} |\p_t \bs n|^2 \\ + \ell_k^2 |\nabla \p_t \bs y|^2 
	\Bigr ) + O(h^{7/2}), \label{eq:kineticexpansion}
\end{gather}
where $\rho_s = h \rho_R$ and the $O(\cdot)$ terms depend only on $C_0, C_1$, $C_2$, $\la$, $\mu$ and $\rho_R$. 
\end{thm}

\begin{proof}
We first note that \eqref{eq:chiexp}, \eqref{eq:Erelations}, and \eqref{eq:strainconst2} imply that
\begin{align}
\forall \bs Y \in \bbR^2, t \in [t_0, t_1], \quad |\p_t \bs g(\bs Y, t)| + |\bs E(\bs Y,t)| \leq \tilde C_2 h, \label{eq:strainconst}
\end{align}
where $\tilde C_2 = \max(1,\sqrt{2} \lambda /(\lambda + 2\mu)) 2C_2$. In what follows, big $O(\cdot)$ terms will depend only on $C_0, C_1$, $C_2$, $\la$, $\mu$ and $\rho$.   

Via Taylor's theorem and \eqref{eq:hellcomp} we have  
\begin{gather}
	\int_{-h/2}^{h/2}  W(\hbF, \hbA) dZ =  \int_{-h/2}^{h/2} \Bigl [ \sum_{n = 0}^2 [\p_Z^n W_s(\hbF, \hbA) |_{Z = 0}] \frac{Z^n}{n!} + O(Z^3) \Bigr ] dZ\\
	+ \int_{-h/2}^{h/2} \Bigl [ \ell_s^2 W_{sg}(\hbF, \hbA) |_{Z = 0} + O(\ell_s^2 Z) \Bigr ] dZ \\
	=  h W_s(\bs F) + h\ell_s^2 W_{sg}(\bs F, \bs A) + \frac{h^3}{24} \p_Z^2 W_s(\hbF) |_{Z = 0} + O(h^4). 
\end{gather}
Direct computation then yields
\begin{align}
	\begin{split}
		\p_Z^2 W_s(\hbF) |_{Z = 0} &= W_{s,\bs F \bs F}(\bs F)[\bs F'] \cdot \bs F' + 
		\bs P_s(\bs F) \cdot \bs F'' \\
		&= W_{s,\bs F \bs F}(\nabla \bs y + \bs d \tens \bs e_3)[\nabla \bs d + {\bs g} \tens \bs e_3]\cdot (\nabla \bs d + \bs g \tens \bs e_3)  \\
		&\quad+ \bs P_s(\nabla \bs y + \bs d \tens \bs e_3) \cdot \nabla \bs g + \bs h \cdot \bs P_s(\nabla \bs y + \bs d \tens \bs e_3)\bs e_3 \\
		&= W_{s,\bs F \bs F}(\nabla \bs y + \bs d \tens \bs e_3)[\nabla \bs d + \bs g \tens \bs e_3]\cdot (\nabla \bs d + \bs g \tens \bs e_3)  \\
		&\quad+ \bs P_s(\nabla \bs y + \bs d \tens \bs e_3) \cdot \nabla \bs g. 
	\end{split}\label{eq:enexpansion}
\end{align}

By \eqref{eq:formofg} and \eqref{eq:bard}, we have for all $\al = 1,2$,
\begin{gather}
	\bs g = \bs L[\nabla \nabla \bs y] + O(h), \quad \p_\al \bs g = \bs L[\nabla \nabla \p_\al \bs y] + O(h), \\ \p_t \bs g = 
	\bs L[\nabla \nabla \p_t \bs y] + O(h), \label{eq:ggrads}
\end{gather}
where $\bs L$ is a constant third-order tensor.  By \eqref{eq:chibound} and \eqref{eq:strainconst} we conclude that $\bs P_s(\nabla \bs y + \bs d \tens \bs e_3) \cdot \nabla \bs g = O(h)$. This fact, \eqref{eq:enexpansion}, and \eqref{eq:dhcomp} show that 
\begin{align}
	\int_{-h/2}^{h/2} W \, dZ = h E_{m} + h \ell_s^2 E_{sg} + \frac{h^3}{24} E_{b} + O(h^4), \label{eq:storedexp}
\end{align}
where 
\begin{align}
	E_m &= W_s(\nabla \bs y + \bs d \tens \bs e_3), \\
	E_{sg} &= W_{sg}(\nabla \bs y  + \bs d \tens \bs e_3, 
	\nabla \nabla \bs y + (\nabla \bs d \tens \bs e_3)^T \\ &\qquad+ \nabla \bs d \tens \bs e_3 + 
	\bs g \tens \bs e_3 \tens \bs e_3), \\
	E_b &= W_{s,\bs F \bs F}(\nabla \bs y + \bs d \tens \bs e_3)[\nabla \bs d + \bs g \tens \bs e_3]\cdot (\nabla \bs d + \bs g \tens \bs e_3).
\end{align} 

Writing 
\begin{align}
	(\tr \bs E)^2 = (\tr \msE)^2 + 2 \tr \msE \, E_{33} + E_{33}^2, \quad 
	|\bs E|^2 = |\msE|^2 + 2 E_{\al 3} E^{\al 3} + E_{33}^2
\end{align}
and inserting \eqref{eq:Erelations}, we obtain 
\begin{align}
	E_m = \frac{\lambda \mu}{\lambda + 2\mu} (\tr \msE)^2 + \mu |\msE|^2. \label{eq:E1almost}
\end{align}

To compute $E_b$ we first note that 
\begin{align}
	\p_\al \bs d = \p_\al \phi \bs n - \phi \sK_{\al \beta} \bs y^{,\beta} = \p_\al \phi \bs n - \sK_{\al \beta} \bs y^{,\beta} + O(h), \label{eq:pald}
\end{align}
where $\{\bs y^{,\beta}\}$ is the dual basis, relative to $\{ \bs y_{,\beta}\}$, that is tangent to $\bs y(\cl S)$. 
Using \eqref{eq:formofg}, it follows that
\begin{align}
	E_b
	= \frac{\p^2 E_m}{\p y^j_{,\beta} \p y^k_{,\alpha}} \p_\al d^k \p_\beta d^j.
\end{align}
(see Section 3 of \cite{HilgPip96}).
By the chain rule and \eqref{eq:strainconst},    
\begin{gather}
	\frac{\p^2 E_m}{\p y^j_{,\beta} \p y^k_{,\alpha}} \bs e^k \tens \bs e^j 
	= \Bigl [ \frac{\lambda \mu}{\lambda + 2\mu} \delta^{\al \rho}\delta^{\beta \gamma} + \frac{\mu}{2}(\delta^{\al \beta} \delta^{\gamma \rho} + \delta^{\al \gamma} \delta^{\beta \rho}) \Bigr ] \bs y_{,\rho} \tens \bs y_{\gamma} + O(h). 
\end{gather}
Then by \eqref{eq:pald}, we conclude 
\begin{align}
	E_b = \frac{\lambda \mu}{\lambda + 2\mu} (\tr \msK)^2 + \mu |\msK|^2 + O(h). 
\end{align}
In particular, we have that  
\begin{gather}
h E_m + \frac{h^3}{24} E_b = h \Bigl (\frac{\lambda \mu}{\lambda + 2\mu} (\tr \msE)^2 + \mu |\msE|^2 \Bigr ) + 
\frac{h^3}{24} \Bigl ( \frac{\lambda \mu}{\lambda + 2\mu} (\tr \msK)^2 + \mu |\msK|^2 \Bigr ) \\ + O(h^4),
\end{gather}
with the first two terms on the right-hand side being Koiter's classical shell energy (see \cite{Koit66, Ciarlet05Book, Steig13}). 

We now split $\sum_{c} (\tr \p_c \bs E)^2$ and $\sum_c |\p_c \bs E|^2$ into parts with and without the index 3. Using \eqref{eq:Erelations} and \eqref{eq:p3Erelations} we compute 
\begin{align}
 \sum_{c = 1}^3 |\p_c \bs E|^2 &= \sum_{\gamma = 1}^2 |\p_\gamma \msE|^2 + 2 \sum_{\gamma = 1}^2 \p_\gamma E_{\al 3} \p_\gamma E^{\al 3} + \sum_{\gamma = 1}^2 (\p_\gamma E_{33})^2 \\
 &\quad + 2 \p_3 E_{\al 3} \p_3 E^{\al 3} + \p_3 E_{\al \beta} \p_3 E^{\al \beta} + (\p_3 E_{33})^2 \\
 &= \sum_{\gamma = 1}^2 |\p_\gamma \msE|^2 + \frac{\lambda^2}{(\lambda + 2\mu)^2} \sum_{\gamma = 1}^2 (\tr \p_\gamma \msE)^2 + \p_3 E_{\al \beta} \p_3 E^{\al \beta}  + (\p_3 E_{33})^2, \\
 \sum_{c = 1}^3 |\tr \p_c \bs E|^2 &= \sum_{\gamma = 1}^2 |\tr \p_\gamma \msE|^2 + 2 \p_\gamma \tr \msE \p_\gamma E_{33} + \sum_{\gamma = 1}^2 (\p_\gamma E_{33})^2 \\
 &\quad + (\p_3 \tr \msE)^2 + 2 \p_3 \tr \msE \p_3 E_{33} + (\p_3 E_{33})^2 \\  
 &= \frac{4 \mu^2}{(\lambda + 2\mu)^2} \sum_{\gamma = 1}^2 (\tr \p_\gamma \msE)^2 + \frac{4 \mu^2}{\lambda^2} (\p_3 E_{33})^2. 
\end{align}
Using \eqref{eq:p3Erelations} and the fact that for $\al, \beta = 1,2$, 
\begin{align}
	\p_3 E_{\al \beta} = \frac{1}{2} \p_3  (\hbchi_{,\al} \cdot \hbchi_{,\beta}) \big |_{Z = 0} = \frac{1}{2} (\p_\al \bs d \cdot \p_\beta \bs y + \p_\al \bs y \cdot \p_\beta \bs d)
	= -\phi \sK_{\al \beta}, \label{eq:curvature}
\end{align}
we conclude that,  
\begin{align}
E_{sg} = \frac{\la \mu}{\la + 2\mu} \sum_{\gamma = 1}^2 (\tr \p_\gamma \msE)^2 + \mu \sum_{\gamma = 1}^2 |\p_\gamma \msE|^2 + \frac{\la \mu}{\la + 2\mu} \phi^2 (\tr \msK)^2 + \mu \phi^2 |\msK|^2, \label{eq:E2almost}
\end{align} 
and thus, 
\begin{gather}
	\int_{-h/2}^{h/2} W \, dZ = 
h \Bigl [ \frac{\lambda \mu}{\lambda + 2\mu} (\tr \msE)^2 + \mu |\msE|^2  
+ \ell_s^2 \sum_{\gamma = 1}^2 \Bigl ( \frac{\lambda \mu}{\lambda + 2\mu} (\tr \p_\gamma \msE)^2 + \mu |\p_\gamma \msE|^2 \Bigr ) \Bigr ] \\
 + 
\Bigl (\frac{h^3}{24} + h \ell_s^2 \Bigr ) \Bigl ( \frac{\lambda \mu}{\lambda + 2\mu} (\tr \msK)^2 + \mu |\msK|^2 \Bigr ) + O(h^4).
\end{gather}
This proves \eqref{eq:storedexpansion}. 

We now consider the kinetic energy. We have by \eqref{eq:chiexp}, \eqref{eq:dhcomp}, \eqref{eq:chibound}, and \eqref{eq:strainconst} that the kinetic energy satisfies  
\begin{gather}
	\int_{-h/2}^{h/2} \frac{1}{2} \rho_R \Bigl (
	|\p_t \hbchi|^2 + \ell_k^2 |\p_t \hbF (\hbF)^{-1}|^2 
	\Bigr ) dZ \\ = \frac{1}{2}\rho_s \Bigl (
	|\p_t \bs y|^2 + \frac{h^2}{12} |\p_t \bs d|^2 + \ell_k^2 |\p_t \bs F (\bs F)^{-1}|^2 + \frac{h^2}{12}
	\p_t \bs y \cdot \p_t \bs g \Bigr )  + O(\ell_k^2 h^3) + O(h^5), \\
	=  \frac{1}{2}\rho_s \Bigl (
	|\p_t \bs y|^2 + \frac{h^2}{12} |\p_t \bs d|^2 + \ell_k^2 |\p_t \bs F (\bs F)^{-1}|^2 \Bigr ) + O(h^4),
	\label{eq:kineticexpproof}
\end{gather}
where $\rho_s = h \rho_R$ is the mass per unit reference area of the midsurface. By \eqref{eq:strainconst} and the polar decomposition theorem, we have 
$\bs F = \bs R(\bs I + O(h^{1/2}))$ where $\bs R$ takes values in the group of rotations, and thus, by \eqref{eq:chiexp} we have  
\begin{gather}
	|\p_t \bs F (\bs F)^{-1}|^2 = |\p_t \bs F \bs R(\bs I + O(h^{1/2}))|^2 = |\p_t \bs F \bs R|^2 + O(h^{1/2}) = |\p_t \bs F|^2 + O(h^{1/2}) \\
	= |\nabla \p_t \bs y|^2 + |\p_t \bs d|^2 + O(h^{1/2}).  
\end{gather}
By \eqref{eq:dhcomp} and \eqref{eq:kineticexpproof}, we conclude that 
\begin{gather}
		\int_{-h/2}^{h/2} \frac{1}{2} \rho_R \Bigl (
	|\p_t \hbchi|^2 + \ell_k^2 |\p_t \hbF (\hbF)^{-1}|^2 
	\Bigr ) dZ = \\
\frac{1}{2}\rho_s \Bigl (
|\p_t \bs y|^2 + \frac{h^2 + 12 \ell_k^2}{12} |\p_t \bs d|^2 + \ell_k^2 |\nabla \p_t \bs y|^2 \Bigr ) + O(h^{7/2}),
\label{eq:kineticexpproof2}
\end{gather}

To simplify \eqref{eq:kineticexpproof2}, we observe that since $\bs n$ is a unit normal vector, we have $\bs n \cdot \p_t \bs n = \frac{1}{2} \p_t |\bs n|^2 = 0$. Thus, by \eqref{eq:bard}
\begin{gather}
	\int_{-h/2}^{h/2} \frac{1}{2} \rho_R \Bigl (
|\p_t \hbchi|^2 + \ell_k^2 |\p_t \hbF (\hbF)^{-1}|^2 
\Bigr ) dZ \\ = \frac{1}{2} \rho_s \Bigl (
		|\p_t \bs y|^2 + \frac{h^2 + 12 \ell_k^2}{12} \frac{\la^2}{(\la + 2\mu)^2} |\tr \p_t \msE|^2 + \frac{h^2+12\ell_k^2}{12} |\p_t \bs n|^2 + \ell_k^2 |\nabla \p_t \bs y|^2 
		\Bigr ), \label{eq:kineticexp}
\end{gather}
concluding the proof. 
\end{proof}

We can express the leading order energies appearing in \eqref{eq:storedexpansion} and \eqref{eq:kineticexpansion}, denoted $U$ and $K$ respectively, in terms of Young's modulus $E = \frac{\mu(2\mu + 3\lambda)}{\mu + \lambda}$ and Poisson's ratio $\nu = \frac{\lambda}{2(\mu+\lambda)} \in (0, 1/2)$. Indeed, we have 
\begin{gather}
U = \frac{a}{2} \Bigl [ \nu (\tr \msE)^2 + (1-\nu) |\msE|^2  
+ \ell_s^2 \sum_{\gamma = 1}^2 \Bigl ( \nu (\tr \p_\gamma \msE)^2 + (1-\nu) |\p_\gamma \msE|^2 \Bigr ) \Bigr ] \\
+ 
\frac{b}{2} \Bigl ( \nu (\tr \msK)^2 + (1-\nu) |\msK|^2 \Bigr ), \label{eq:surfacestored}
\end{gather}
and 
\begin{gather}
K = 
\frac{1}{2} \rho_s \Bigl (
|\p_t \bs y|^2 + \frac{\nu^2}{(1-\nu)^2} c |\tr \p_t \msE|^2 + c |\p_t \bs n|^2 + \ell_k^2 |\nabla \p_t \bs y|^2 
\Bigr ) \label{eq:surfacekinetic}
\end{gather}
where $\rho_s = h \rho_R$ and 
\begin{align}
	a = \frac{h E}{1-\nu^2}, \quad b = \frac{h E}{1-\nu^2} \Bigl (\frac{h^2}{24} + \ell_s^2 \Bigr ), \quad c = \frac{h^2 + 12 \ell_k^2}{12}, \label{eq:abc}
\end{align}
are all positive. Going forward we will utilize the expressions \eqref{eq:surfacestored} and \eqref{eq:surfacekinetic} due to their more compact form. If we identify the length scales $\ell_s$ and $\ell_k$ via \eqref{eq:specificls}, then the surface energies \eqref{eq:surfacestored} and \eqref{eq:surfacekinetic} are parameterized by five physical properties for a given material: its reference density $\rho_R$, Young's modulus $E$, Poisson's ratio $\nu$, thickness $h$, and natural configuration's inter-particle spacing $d$.  

\subsection{Strong ellipticity of the stored surface energy} We now show that the surface energy $U$ defined as in \eqref{eq:surfacestored} satisfies the following \emph{strong ellipticity condition} as long as $\ell_s > 0$. 
\begin{prop}\label{p:p1}
Let $U$ be as in \eqref{eq:surfacestored} with $E > 0$ and $\nu \in (0,1/2)$, and let $\bs y : \cl S \rar \bbE^3$ be an immersion. Let 
\begin{align}
		\bs C^{\al \beta \gamma \delta}(\bs Y) := \frac{\p^2 U}{\p y^{i}_{,\al \beta} \p y^{j}_{,\de \gamma}} \Big |_{\bs y(\bs Y)} \bs e^i \tens \bs e^j. 
\end{align} 
If $\ell_s > 0$, then $U$ satisfies the strong ellipticity condition: for all $\bs Y \in \cl S$, $(\sa_1, \sa_2) \in \bbR^2 \backslash\{\bs (0,0)\}$, and $\bs b \in \bbR^3 \backslash \{\bs 0\}$,
\begin{align}
	\sa_{\al} \sa_{\beta} \bs b \cdot \Bigl ( \bs C^{\al \be \delta \gamma}(\bs Y) \sa_{\de} \sa_{\gamma} \bs b \Bigr ) > 0  \label{eq:SE}.
\end{align}
If $\ell_s = 0$, then $U$ satisfies the weaker Legendre-Hadamard condition: for all $\bs Y \in \cl S$, $(\sa_1, \sa_2) \in \bbR^2 \backslash\{\bs (0,0)\}$, and $\bs b \in \bbR^3 \backslash \{\bs 0\}$,
\begin{align}
	\sa_{\al} \sa_{\beta} \bs b \cdot \Bigl ( \bs C^{\al \be \delta \gamma} \sa_{\de} \sa_{\gamma}(\bs Y) \bs b \Bigr ) \geq 0  \label{eq:LH},
\end{align}
and the left-hand size of \eqref{eq:LH} is zero precisely when $\bs b \cdot \bs n(\bs Y) = 0$. 
\end{prop}

\begin{proof} 
Via straightforward calculations, we have the relations  
\begin{equation}
\label{eq:relations}
\begin{gathered}
	\frac{\p \sE_{\beta \nu}}{\p \bs y_{,\al}} = \frac{1}{2}\Bigl ( \delta^\al_\beta \bs y_{,\nu} + \delta^\al_{\nu} \bs y_{,\beta} \Bigr ), \\
	\frac{\p(\p_\nu \sE_{\al \beta})}{\p \bs y_{,\rho}} = \frac{1}{2} \Bigl (\delta^\rho_\beta \bs y_{,\al \nu} + \delta^\rho_\al \bs y_{,\beta \nu} \Bigr ), \\ 
	\frac{\p(\p_\nu \sE_{\al \beta})}{\p \bs y_{,\rho \sigma}} = \frac{1}{4} \Bigl (\delta^\rho_\al \delta^\sigma_\nu + \delta^\rho_\nu \delta^\sigma_\al \Bigr ) \bs y_{,\beta} + \frac{1}{4} \Bigl (\delta^\rho_\beta \delta^\sigma_\nu + \delta^\rho_\nu \delta^\sigma_\beta \Bigr ) \bs y_{,\al}, \\
	\frac{\p \sK_{\al \beta}}{\p \bs y_{,\nu}} = -\tensor{\gamma}{^\nu_\al_\beta} \bs n, \quad 
	\frac{\p \sK_{\al \beta}}{\p \bs y_{,\rho \sigma}} = \frac{1}{2} \Bigl ( \delta^\rho_\al \delta^\sigma_\beta + \delta^\sigma_\beta \delta^\rho_\al \Bigr ) \bs n,
\end{gathered}
\end{equation}
where $\sg_{\al \beta} = \bs y_{,\al} \cdot \bs y_{,\beta}$ are the components of the metric tensor on the convected surface and $\tensor{\gamma}{^\nu_\al_\beta}$ are the Christoffel symbols associated to the metric $\bs \msg$. These relations lead to the identities 
\begin{gather}
	\frac{\p^2}{\p \bs y_{,\al \beta} \p \bs y_{, \gamma \rho}} (\tr \msK)^2 = 2 \delta^{\al \beta} \delta^{\gamma \rho} \bs n \tens \bs n, \\
	\frac{\p^2}{\p \bs y_{,\al \beta} \p \bs y_{, \gamma \rho}} |\msK|^2 = (\delta^{\al \gamma}\delta^{\beta \rho} + \delta^{\al \rho} \delta^{\beta \gamma}) \bs n \tens \bs n, \\
	\frac{\p^2}{\p \bs y_{,\al \beta} \p \bs y_{, \gamma \rho}} \sum_\zeta (\p_\zeta \tr \msE)^2 = \frac{1}{2}
	(\delta^{\al \eta} \delta^{\beta \zeta} + \delta^{\al \zeta} \delta^{\beta \eta})
	(\delta^{\gamma \theta} \delta^\rho_\zeta + \delta^\gamma_\zeta \delta^{\rho \theta}) \bs y_{,\eta} \tens \bs y_{,\theta},
\end{gather} 
and 
\begin{gather}
	\frac{\p^2}{\p \bs y_{,\al \beta} \p \bs y_{, \gamma \rho}} \sum_\zeta |\p_\zeta \tr \msE|^2 = 
	\frac{1}{2}(\delta^{\gamma \beta} \delta^{\rho \alpha} + \delta^{\gamma \al} \delta^{\rho \beta}) \delta^{\eta \theta} \bs y_{,\eta} \tens \bs y_{,\theta} \\
		+ \frac{1}{4} \bigl [
		(\delta^{\gamma \beta} \delta^{\rho \eta} + \delta^{\gamma \eta} \delta^{\rho \beta}) \delta^{\al \theta}
		+ (\delta^{\gamma \al} \delta^{\rho \eta} + \delta^{\gamma \eta} \delta^{\rho \al}) \delta^{\beta \theta}
		\bigr ] \bs y_{,\eta} \tens \bs y_{,\theta}.
\end{gather}
For $\bs b \in \bbR^3$, let $\sab_\al = \bs b \cdot \bs y_{,\al}$. Then the previous imply that for all $(\sa_1, \sa_2) \in \bbR^2 \backslash\{\bs (0,0)\}$, and $\bs b \in \bbR^3 \backslash \{\bs 0\}$, 
\begin{gather}
	\sa_{\al} \sa_{\beta} \bs b \cdot \Bigl ( \bs C^{\al \be \delta \gamma} \sa_{\de} \sa_{\gamma} \bs b \Bigr ) =
	\frac{a}{2} \ell_s^2 \Bigl [ (1 - \nu) |\msa|^2 \sum_\al \sab_\al^2 + (1+\nu)|\msa|^2 \Bigl ( \sum_\al
	\sa_\al \sab_\al
	\Bigr )^2
	\Bigr ] \\
	+ b |\msa|^2 (\bs b \cdot \bs n)^2. \label{eq:Ccomputation}
\end{gather}
The conclusions of the proposition then follow immediately from \eqref{eq:Ccomputation} and the facts that $a,b > 0$, $\nu \in (0,1/2)$. 
\end{proof}

\section{Surface dynamics associated to the energies $U$ and $K$}

In this section we study the dynamics of an elastic material planar surface with stored surface energy density $U$ and kinetic energy density $K$, defined in \eqref{eq:surfacestored} and \eqref{eq:surfacekinetic} respectively.  

\subsection{Field equations}

The field equations governing the motion of a material planar surface $\cl S \subseteq \bbR^2$ with surface energy density $U$ and kinetic energy density $K$ are derived via Hamilton's variational principle, summarized as follows (see also \cite{Hilgers97}). 
Suppose that $\cl P \subseteq \cl S$ is the closure of a domain with $\cl C = \p \cl P$ given by a smooth closed curve parameterized by arclength $S$, with unit tangent $\bs \zeta$, and with outward normal $\bs \eta$. For $[t_0, t_1] \subseteq [0,\infty)$ and a motion $\bs y(\cdot,t)$, the \emph{kinetic energy} of the part $\cl P$ is $T(\cl P; t) = \int_{\cl P} K \, dA$ where $K$ is given by \eqref{eq:surfacekinetic}  \emph{stored energy} of the part $\cl P$ is $V(\cl P; t) = \int_{\cl P} U \, dA$ where $U$ is given by \eqref{eq:surfacestored}.  
In what follows, we denote
\begin{gather} 
	\msT^\al = \frac{\p U}{\p y^i_{,\al}} \bs e^i, \quad \msM^{\al \beta} = \frac{\p U}{\p y^i_{,\al \beta}} \bs e^i,  \quad 
	\bs \Pi^\al = \frac{\p K}{\p \dot{y}^i_{,\al}} \bs e^i, \\
  \msP^\al = \msT^\al - \p_{\be} \msM^{\al \beta} + \frac{\p}{\p t} \bs \Pi^\al. \label{eq:stressvectors}
\end{gather}

Throughout the remainder of this work, we denote $$\p^\beta = \delta^{\beta \rho} \p_\rho.$$ Using the relations \eqref{eq:relations} and $|\p_t \bs n|^2 = (\sg^{-1})^{\beta \nu} (\bs n \cdot \p_t \bs y_{,\beta}) (\bs n \cdot \p_t \bs y_{,\nu}),$ we obtain
\begin{gather}
\msT^\al = a \Bigl ( \nu (\tr \msE) \delta^{\al \beta} + (1-\nu) \sE^{\al \beta} \Bigr ) \bs y_{,\beta}
+ a \ell_s^2 \Bigl ( \nu (\tr \p^\rho \msE) \delta^{\al \beta} + (1-\nu) \p^\rho \sE^{\al \beta} \Bigr ) \bs y_{,\beta \rho} \\
- b \Bigl (
\nu (\tr \msK) \delta^{\beta \rho} \tensor{\gamma}{^\al_\beta_\rho} + (1-\nu) \sK^{\beta \rho} \tensor{\gamma}{^\al_\beta_\rho}
\Bigr ) \bs n,\\
\msM^{\al \beta} = \frac{a}{2} \ell_s^2 \Bigl [ \nu \Bigl (
 (\tr \p^\beta \msE) \delta^{\al \rho} +  (\tr \p^\al \msE) \delta^{\beta \rho}
\Bigr ) + (1-\nu) \Bigl (\p^\al \sE^{\beta \rho} + \p^\beta \sE^{\al \rho} \Bigr ) \Bigr ] \bs y_{,\rho} \\
+ b \Bigl (
\nu (\tr \msK ) \delta^{\al \beta} + (1-\nu) \sK^{\al \beta} \Bigr )
\bs n, \\
\bs \Pi^\al = \rho_s \Bigl [ \frac{\nu^2}{(1-\nu)^2}c (\tr \p_t \msE) \delta^{\al \beta} \bs y_{,\beta} + 
c (\sg^{-1})^{\al \beta}(\bs n \cdot \p_t \bs y_{,\beta}) \bs n \Bigr ] \\
+ \rho_s \ell_k^2 \delta^{\al \beta} \p_t \bs y_{,\beta}. 
\end{gather}

The \emph{action} of the motion in $\cl P \times [t_0, t_1]$ is 
$
	\cl A(\cl P \times [t_0, t_1]; \bs \chi) = \int_{t_0}^{t_1} [
	T(\cl P; t) - V(\cl P; t) 
	 ] dt. 
$
Let $\bs y_{\eps}$ be a smooth one parameter family of deformations of $\cl S$ such that $\bs y_0 = \bs y$ and $\bs \psi := \frac{d}{d\eps} \bs y_{\eps} |_{\eps = 0} : \cl S \times [t_0, t_1] \rar \bbR^3$ is a smooth variation. The field equations governing the motion of $\cl S$ are the Euler-Lagrange equations associated to the variational equation: for all $\bs \psi$, 
\begin{gather}
	\frac{d}{d\eps} \cl A(\cl P \times [t_0, t_1]; \bs y_{\eps} ) \big |_{\eps = 0} + \int_{t_0}^{t_1} \int_{\cl P} \bs f \cdot \bs \psi \, dA dt \\
	+\int_{t_0}^{t_1} \int_{\cl C} \Bigl (\bs \tau \cdot \bs \psi + \bs \mu \cdot (\eta^\beta \bs \psi_{,\beta} \Bigr ) dS dt - \int_{\cl P} \Bigl ( \rho_s \p_t \bs y \cdot \bs \psi + \bs \Pi^\al \cdot \bs \psi_{,\al} \Bigr ) dA \Big |_{t_0}^{t_1} 
	= 0, \label{eq:vareq}  
\end{gather}
where $\bs f$ is a prescribed external body force on $\cl S$, and $\bs \tau$ and $\bs \mu$ are generalized tractions. We refer the reader to \eqref{eq:linearmomentum} and \eqref{eq:angularmomentum} and their follow-up comments below for the interpretation of $\bs \tau$ and $\bs \mu$ in terms of resultant contact forces and couples, respectively.

Using the chain rule we compute   
\begin{gather}
	\frac{d}{d\eps} \cl A(\cl P \times [t_0, t_1]; \bs y_{\eps} ) \big |_{\eps = 0} = 
	\int_{t_0}^{t_1} \int_{\cl P} \Bigl ( \rho \p_t \bs y \cdot \p_t \bs \psi + \bs \Pi^\al \cdot \p_t \bs \psi_{,\al} \Bigr ) dA dt \\ 
	- \int_{t_0}^{t_1} \int_{\cl P} \Bigl (
	\msT^\al \cdot \bs \psi_{,\al} + \msM^{\al \beta} \bs \psi_{,\al \beta}
	\Bigr ) dA dt \label{eq:actionvar}
\end{gather}
Via straightforward calculations repeatedly using the divergence theorem, we then conclude that \eqref{eq:vareq} is satisfied for all variations if and only if: 
\begin{align}
	\begin{split}
		\rho_{s} \p_t^2 \bs y &= \p_{\al} \msP^\al + \bs f, \quad \mbox{ on } \cl P \times [t_0, t_1], \\
			\bs \tau &= \msP^\al \eta_\al - \frac{\p}{\p S} 
			\Bigl (\msM^{\al \beta} \zeta_\al \eta_\beta \Bigr ), \quad \mbox{ on } \p \cl P \times [t_0, t_1], \\
			\bs \mu &= \msM^{\al \beta} \eta_\al \eta_\beta, \quad \mbox{ on } \p \cl P \times [t_0, t_1].    
	\end{split}\label{eq:field}
\end{align}

\subsection{Balance laws}
By choosing the variation to be appropriate infinitesimal generators of spatial translations, spatial rotations, and temporal translations, we obtain standard balance laws for the part $\cl P$, expressed in the reference configuration. 

Indeed, let $\bs a \in \bbR^3$ and $\bs y_\eps(\bs Y,t) = \bs y(\bs Y,t) + \eps \bs a$ so $\bs \psi = \bs a$. Since the action's Lagrangian is clearly invariant with respect to superimposed (constant) spatial translations of $\bs y$, we conclude that $\frac{d}{d\eps} \cl A(\cl P \times [t_0, t_1]; \bs y_\eps) = 0$. Then \eqref{eq:vareq} implies that 
\begin{gather}
	\int_{t_0}^{t_1} \int_{\cl P} \bs f \cdot \bs a \, dA dt  
	+\int_{t_0}^{t_1} \int_{\cl C} \bs \tau \cdot \bs a \, dS dt - \int_{\cl P} \rho_s \p_t \bs y \cdot \bs a \, dA \Big |_{t_0}^{t_1} 
	= 0.
\end{gather}
Dividing the previous by $t_1 - t_0$ and taking the limit $t_1 \rar t_0$ yields 
\begin{align}
	\frac{d}{dt} \int_{\cl P} \rho_s \p_t \bs y \cdot \bs a \, dA =  \int_{\cl P} \bs f \cdot \bs a \, dA + 
	\int_{\cl C} \bs \tau \cdot \bs a \, dS. \label{eq:linmoma}
\end{align}
The equation \eqref{eq:linmoma} holding for all $\bs a \in \bbR^3$ implies the following relation interpreted as the \emph{balance of linear momentum} for the part $\cl P$:  
\begin{align}
	\frac{d}{dt} \int_{\cl P} \rho_s \p_t \bs y \, dA &=  \int_{\cl P} \bs f \, dA 
	+\int_{\cl C} \bs \tau \, dS. \label{eq:linearmomentum}
\end{align}
From \eqref{eq:linearmomentum} we see that the generalized traction $\bs \tau$ contributes a resultant contact force exerted on $\cl P$ by $\cl S \backslash \cl P$ through the term $\int_{\cl C} \bs \tau \, dS$. 

Now, let $\bs \Omega = \bs a \times$ be an arbitrary skew symmetric tensor with axial vector $\bs a$ and let $\bs y_\eps(\bs Y, t) = e^{\eps \bs \Omega} \bs y(\bs Y, t)$. Then $\bs \psi = \bs a \times \bs y$. We note that the strain tensors $\msE$, $\msK$ and $\nabla \msE$ are invariant with respect to super-imposed rotations of $\bs y$, and thus, $\frac{d}{d\eps} \cl A(\cl P \times [t_0, t_1]; \bs y_\eps) = 0$. Then \eqref{eq:vareq} and the circularity of the scalar triple product imply that
\begin{gather}
	\int_{t_0}^{t_1} \int_{\cl P}(\bs y \times \bs f) \cdot \bs a \, dA dt   
	+\int_{t_0}^{t_1} \int_{\cl C} \bigl ( \bs y \times \bs \tau + \eta^\nu \bs y_{,\nu} \times \bs \mu \bigr ) \cdot \bs a \, dS dt \\
	- \int_{\cl P} \rho_s \bigl (
	\bs y \times \p_t \bs y + 
	\bs y_{,\al} \times \bs \Pi^\al  \bigr ) \cdot \bs a \, dA \Big |_{t_0}^{t_1}
	= 0.
\end{gather}
As before, the previous implies \emph{balance of angular momentum} for the part $\cl P$: 
\begin{gather}
	\frac{d}{dt} 
	\int_{\cl P} \rho_s \bigl (
	\bs y \times \p_t \bs y + \bs y_{,\al} \times \bs \Pi^\al \bigr ) dA
	\\=  \int_{\cl P}\bs y \times \bs f \, dA 
	+ \int_{\cl C} \bigl ( \bs y \times \bs \tau + \eta^\nu \bs y_{,\nu} \times \bs \mu \bigr )  dS. \label{eq:angularmomentum}
\end{gather}
Referring to \eqref{eq:angularmomentum}, we can observe that $\bs{\Pi}^\alpha$ introduces a non-classical component, represented by the integral $\int_{\mathcal{P}} \bs{y}_{,\alpha} \times \bs{\Pi}^\alpha \, dA$, into the angular momentum. Additionally, the generalized traction $\bs{\mu}$ contributes a resultant contact couple applied to $\mathcal{P}$ by $\mathcal{S}\backslash \cl P$ through the term $\int_\mathcal{C} \eta^\nu \bs{y}_{,\nu} \times \bs{\mu} \, dS$.

Finally, choosing $\bs y_\eps(\bs Y, t) = \bs y(\bs Y, t+\eps)$ and arguing similarly as above lead to \emph{balance of energy} for the part $\cl P$:
\begin{align}
	\frac{d}{dt} \Bigl (T(\cl P; t) + V(\cl P; t) \Bigr ) = \int_{\cl P} \bs f \cdot \p_t \bs y \, dA 
	+\int_{\cl C} \Bigl (\bs \tau \cdot \p_t \bs y + \bs \mu \cdot (\eta^\nu \p_t \bs y_{,\nu}) \Bigr ) dS.
\end{align}

\subsection{Plane harmonic waves for infinitesimal displacements}

Let $\bs u(\bs Y, t) = \bs y(\bs Y,t) - \bs Y$ be the displacement field for the material surface. We write 
\begin{align}
	\bs u = \msu + w \bs e_3, \quad \msu = \su^\al \bs e_\al, \quad \eps_{\al \beta} = \frac{1}{2}(\p_\al \su_\beta + \p_\beta \su_\al). 
\end{align} 
Assuming that $|\bs u|$, $|\p_t \bs u| + |\nabla \bs u|$, and $|\nabla \p_t \bs u| + |\nabla \nabla \bs u|$, are bounded by $\delta \ll 1$, we have
\begin{gather}
\sE_{\al \beta} = \eps_{\al \beta} + O(\delta^2), \quad \sK_{\al \beta} = \p_\al \p_\beta w + O(\delta^2), \\
	\msT^\al = a \Bigl ( \nu (\tr \bs \eps) \delta^{\al \beta} + (1-\nu) \eps^{\al \beta} \Bigr ) \bs e_{\beta} + O(\delta^2), \\
\msM^{\al \beta} = \frac{1}{2} \ell_s^2 a \Bigl [ \nu \Bigl (
(\tr \p^\beta \bs \eps) \delta^{\al \rho} +  (\tr \p^\al \bs \eps) \delta^{\beta \rho}
\Bigr ) + (1-\nu) \Bigl (\p^\al \eps^{\beta \rho} + \p^\beta \eps^{\al \rho} \Bigr ) \Bigr ] \bs e_{\rho} \\
+ b \Bigl (
\nu (\p_\nu \p^\nu w ) \delta^{\al \beta} + (1-\nu) \p^\al \p^\beta w \Bigr )
\bs e_3 + O(\delta^2), \\
\bs \Pi^\al = \rho_s \Bigl [
\frac{\nu^2}{(1-\nu)^2} c (\tr \p_t \bs \eps) \delta^{\al \beta} \bs e_{\beta} + 
c \delta^{\al \beta}(\p_\beta \p_t w) \bs e_3 \Bigr ] \\
+ \rho_s \ell_k^2 \delta^{\al \beta} \p_\beta \p_t \bs u + O(\delta^2).
\end{gather}
The equations governing infinitesimal displacements correspond to the linearization of \eqref{eq:field}, i.e., we drop all terms that are $O(\delta^2)$, leading to: 
\begin{align}
\begin{split}
	\rho_s \p_t^2 \bs u = (1 - \ell_s^2 \p_\rho \p^\rho) a \Bigl (
	\nu \p^\beta \tr \bs \eps + (1-\nu) \p_\al \eps^{\al \beta} 
	\Bigr ) \bs e_\beta
- b (\p_\al \p^\al)^2 w \bs e_3 \\
+ \rho_s \Bigl (
\frac{\nu^2}{(1-\nu)^2} c \p^\beta \p_t^2 \tr \bs \eps + \ell_k^2 \p_\al \p^\al \p_t^2 \su^\beta \Bigr ) \bs e_\beta
+ \rho_s  (c + \ell_k^2 )\p_\al \p^\al \p_t^2 w \bs e_3 + \bs f. 
\end{split}\label{eq:linindices}
\end{align}
In terms of $\msu$, $w$, and standard notation wherein differential operators are with respect to the variables $\bs x = (x^1, x^2) = (Y^1, Y^2)$, we may express \eqref{eq:linindices} as 
\begin{align}
\begin{split}
	\rho_s \Bigl ( (1-\ell_k^2 \Delta) \p_t^2 \msu - \frac{\nu^2}{(1-\nu)^2} c \nabla \div \p_t^2 \msu \Bigr ) &=  \div a (1-\ell_s^2 \Delta) (
	\nu(\tr \bs \eps)\msP + (1-\nu) \bs \eps
	 ) \\ &\qquad +\msf, \\
	\rho_s \Bigl (\p_t^2 w -  (c + \ell_k^2 ) \Delta \p_t^2 w \Bigr ) &= - b \Delta^2 w + l, 
\end{split}\label{eq:limfam}
\end{align}
where we have written $\bs f = \msf + l \bs e_3$, $\msf = f^\al \bs e_\al$, and $\msP = \bs e_1 \tens \bs e_1 + \bs e_2 \tens \bs e_2$. Now viewing $\cl S$ as the midsurface of a gradient elastic plate with Young's modulus $E$, Poisson's ratio $\nu$, thickness $h$, and inter-particle spacing $d$, we have the identifications 
\begin{align}
		a = \frac{h E}{1-\nu^2}, \quad b = \frac{h E}{1-\nu^2} \Bigl (\frac{h^2}{24} + \ell_s^2 \Bigr ), \quad c = \frac{h^2}{12} + \ell_k^2. 
\end{align} 

We now consider a plane harmonic wave 
\begin{align}
	\bs u = A \exp[i(\msk \cdot \bs x - \omega t)] \bs d, \label{eq:planeharmonic}
\end{align}
where $\bs x = x^\mu \bs e_\mu$, $\msk = k^\mu \bs e_\mu \neq \bs 0$ is the wave vector, $\bs d = d^i \bs e_i$ is the unit length direction of motion, $\omega$ is the angular frequency, and $c_s = \omega/|\msk|$ is the phase velocity. We denote 
\begin{align}
	\hat{\msk} = \frac{\msk }{|\msk|}, \quad \msk^\perp = -k^2 \bs e_1 + k^1 \bs e_2, \quad \msd = d^\al \bs e_\al.  
\end{align}
By inserting \eqref{eq:planeharmonic} into \eqref{eq:limfam} we conclude that 
\begin{align}
\omega^2 \rho_s \Bigl [ (1 + \ell_k^2|\msk|^2) \msd &+ c \frac{\nu^2}{(1-\nu)^2} (\msk \cdot \bs d) \msk \Bigr ] \\ &= a(1 + \ell_k^2 |\msk|^2 )  \Bigl [
\frac{1+\nu}{2} (\msk \cdot \bs d) \msk + \frac{1-\nu}{2} |\msk|^2 \msd \Bigr ], \\	
\omega^2 \rho_s (1 + (c + \ell_k^2)|\msk|^2) (\bs d \cdot \bs e_3) &= {b}|\msk|^4 (\bs d \cdot \bs e_3).
\end{align}
The previous are equivalent to 
\begin{align}
	\bs A(\msk) \bs d = c_s^2 \bs d, \label{eq:dequation}
\end{align} 
where $\bs A(\msk)$ is the acoustical tensor defined by
\begin{align}
\begin{split}
\bs A(\msk) &= 
\frac{a}{\rho_s} \frac{1+\ell_s^2|\msk|^2}{1 + (c\nu^2(1-\nu)^{-2} + \ell_k^2)|\msk|^2} \hat{\msk} \tens 
\hat{\msk} + \frac{a(1-\nu)}{2\rho_s} \frac{1+\ell_s^2|\msk|^2}{1 + \ell_k^2|\msk|^2} \hat{\msk}^{\perp} \tens \hat{\msk}^{\perp} \\
 &\qquad + \frac{b}{\rho_s} \frac{|\msk|^2}{1+(c + \ell_k^2)|\msk|^2} \bs e_3 \tens \bs e_3. 
\end{split}\label{eq:acoustic}
\end{align}

From \eqref{eq:dequation} and \eqref{eq:acoustic} we conclude that for a given nonzero wave vector $\msk$, there are three types of plane harmonic waves corresponding to the three linearly independent directions of motion:
\begin{itemize}
\item There is one longitudinal wave with direction of motion $\bs d = \hat{\bs k}$ and phase velocity satisfying
\begin{align}
	c_{s,L}^2(|\msk|) = \frac{a}{\rho_s} \frac{1+\ell_s^2|\msk|^2}{1 + (c\nu^2(1-\nu)^{-2} + \ell_k^2)|\msk|^2}. 
\end{align} 
\item There is one tangentially transverse wave with direction of motion $\bs d = \hat{\bs k}^{\perp}$ and phase velocity satisfying
\begin{align}
	c_{s,T}^2(|\msk|) = \frac{a(\nu-1)}{2\rho_s} \frac{1+\ell_s^2|\msk|^2}{1 + \ell_k^2|\msk|^2}. 
\end{align}   
\item There is one normally transverse wave with direction of motion $\bs d = \bs e_3$ and phase velocity satisfying
\begin{align}
	c_{s,N}^2(|\msk|) = \frac{b}{\rho_s} \frac{|\msk|^2}{1+(c + \ell_k^2)|\msk|^2}. 
\end{align}
\end{itemize}

In our remaining discussion of the effects that the length parameters have on plane harmonic waves, we adopt the identifications from \eqref{eq:specificls}, i.e., 
\begin{align}
	\ell_s^2 = \frac{d^2}{12}, \quad \ell_k^2 = \frac{d^2}{6}.
\end{align}
We denote the classical phase velocities where $\bs \ell_s = \bs \ell_k = 0$ by
\begin{gather}
	c_{s,L,cl}^2(|\msk|) = \frac{a}{\rho_s} \frac{1}{1 + \frac{h^2}{12} \nu^2(1-\nu)^{-2}|\msk|^2}, \\ 
	c_{s,T,cl}^2(|\msk|) = \frac{a(\nu-1)}{2\rho_s}, \quad 
	c_{s,N,cl}^2(|\msk|) = \frac{a}{\rho_s}\frac{\frac{h^2}{24} |\msk|^2}{1 + \frac{h^2}{12}|\msk|^2}
\end{gather}  
and we point out the degenerate property of classical longitudinal waves that the phase velocity vanishes in the short wavelength limit,  
\begin{align}
	\lim_{k \rar \infty} c_{s,L,cl}^2(k) = 0. \label{eq:deflongwave} 
\end{align} 

Now, when it comes to tangentially transverse waves, things change from the classical setting. The phase velocities of these waves show dispersion and are consistently \emph{slower} than their classical counterparts: for all $|\msk| \neq 0$,
\begin{align}
	c_{s,T}^2(|\msk|) < c_{s,T,cl}^2(|\msk|).
\end{align} 

However, normally transverse waves behave differently. When $|\msk|^2 < \frac{12}{h^2}$, their phase velocities are \emph{faster} than the classical ones, $c_{s,N}^2(|\msk|) > c_{s,N,cl}^2(|\msk|)$, and when $|\msk|^2 < \frac{12}{h^2}$ their phase velocities are \emph{slower} than the classical ones, $c_{s,N}^2(|\msk|) < c_{s,N,cl}^2(|\msk|)$. 

Longitudinal waves also display a form of threshold behavior. When the wave number magnitude satisfies $	|\msk|^2 < \frac{12}{h^2}( 2 +  (\frac{1}{\nu} - 1 )^2  )$, then the phase velocity is \emph{slower} than the classical ones, $c_{s,N}^2(|\msk|) < c_{s,N,cl}^2(|\msk|)$. When the wave number magnitude is above this threshold, $|\msk|^2 > \frac{12}{h^2} ( 2 +  (\frac{1}{\nu} - 1 )^2 )$, then the phase velocity is \emph{faster} than the classical ones, $c_{s,N}^2(|\msk|) > c_{s,N,cl}^2(|\msk|)$.   
Moreover, in stark contrast to the classical setting, the phase velocity is positive in the short wavelength limit:
\begin{gather}
	\lim_{k \rar \infty} c_{s,L}^2(k) = \frac{a}{\rho_s} \frac{\ell_s^2}{c\nu^2(1-\nu)^{-2} + \ell_k^2} > 0. 
\end{gather}

\section{Using the stored surface energy to model crack fronts}

In this final section, we briefly demonstrate how utilizing the stored energy density $U$ in the framework proposed by the author in \cite{Rodriguez2023StrainGradient} eliminates the singularities in stresses and strains observed in linear elastic fracture mechanics, particularly in the context of mode-III fracture. The overall model can be physically understood as representing a body primarily made of a Green elastic material with an additional thin, gradient elastic region of thickness $h$, extending from a section of its boundary. More precisely, we model such a body by a Green elastic solid $\tilde{\cl B}$ possessing a gradient elastic boundary surface, $\tilde{\cl S}$, with stored surface energy density $U$; see Figure \ref{f:1} from Section 1 for a schematic. The fracture model is parameterized by the three-dimensional solid's Young's modulus, Poisson's ratio, and natural inter-particle spacing. The only remaining factor to determine is the width of the region near the crack front where small-scale gradient elastic effects become significant.

\subsection{Linearized equations for infinitesimal displacements}

We recall from \cite{Rodriguez2023StrainGradient} the following set-up. Consider a Green elastic body $\tilde{\cl B} \subset \bbE^3$, possessing a stored energy density $\tilde W$. This body also contains a gradient elastic boundary surface $\tilde{\cl S} \subseteq \p \tilde{\cl B}$.  For our purposes we will also assume that $\tilde{\cl S} \subseteq \R^2$. The surface $\tilde{\cl S}$ has an associated stored surface energy density, denoted as $\tilde U$. Let $\bs f$ be an external body force, $\bs t$ a prescribed traction on $\tilde{\cl S}$, and $\bs \chi_0$ a prescribed placement of $\p \tilde{\cl B} \backslash \tilde{\cl S}$. The field equations governing equilibrium configurations $\bs \chi : \tilde{\cl B} \rar \bbE^3$ are given by
\begin{align}
	\begin{split}
		&\Div \tilde{\bs P} + \bs f = \bs 0, \quad \mbox{on } \tilde{\cl B}, \\
		&\tilde{\bs P} \bs N = \p_\al \tilde{\msP}^\al + \bs t, \quad \mbox{on } \tilde{\cl S}, \\
		&\bs \chi = \bs \chi_0, \quad \mbox{on } \p \tilde{\cl B} \backslash \tilde{\cl S}, 
	\end{split}\label{eq:fieldequations}
\end{align}
where $\bs N$ is the outward-pointing normal vector field on $\tilde{\cl S}$, $\tilde{\bs P} = \tensor{\tilde P}{_i^a} \bs e^i \tens \bs e_a$ is the Piola stress with $\tensor{\tilde P}{_i^a} = \frac{\p \tilde W}{\p \tensor{F}{^i_a}}$, $\Div \tilde{\bs P} = \bigl ( \p_{X^a} \tensor{\tilde P}{_i^a} \big ) \bs e^i$, and 
\begin{gather} 
		\tilde{\msT}^\al = \frac{\p \tilde U}{\p y^i_{,\al}} \bs e^i, \quad \tilde{\msM}^{\al \beta} = \frac{\p \tilde U}{\p y^i_{,\al \beta}} \bs e^i,  \quad
		\tilde{\msP}^\al = \tilde{\msT}^\al - \p_{\be} \tilde{\msM}^{\al \beta}.
\end{gather}
See Section 2 of \cite{Rodriguez2023StrainGradient}.

Now we assume that $\tilde{\cl B}$ is both homogeneous and isotropic, characterized by a Young's modulus $\tilde E$ and Poisson's ratio $\tilde \nu$. Additionally, we make the assumption that $\tilde U = U$, where $U$ is defined in equation \eqref{eq:surfacestored}. The equilibrium equations for infinitesimal displacements $\bs u : \tilde{\cl B} \rar \bbR^3$ correspond to the linearization of \eqref{eq:fieldequations}. By results of Sections 4.1 and 4.3, these equations are given by 
\begin{align}
	\begin{split}
		&\Div \bs \sigma + \bs f = \bs 0, \quad \mbox{on } \tilde{\cl B}, \\
		&\bs \sigma \bs N = (1 - \ell_s^2 \p_\rho \p^\rho) a \Bigl (
		\nu \p^\beta \tr \bs \eps + (1-\nu) \p_\al \eps^{\al \beta} 
		\Bigr ) \bs e_\beta \\
		&\qquad - b (\p_\al \p^\al)^2 w \bs e_3 + \bs t, \quad \mbox{on } \tilde{\cl S}, \\
		&\bs u = \bs 0, \quad \mbox{on } \p \tilde{\cl B} \backslash \tilde{\cl S}
	\end{split}\label{eq:linfieldequations}
\end{align} 
where $\bs \sigma = \frac{\tilde E}{1+\tilde \nu}(\frac{\tilde \nu}{1-2\tilde \nu} (\tr \bs \varepsilon) \bs I + \bs \varepsilon)$ and $\varepsilon_{ij} = \frac{1}{2}(\p_{i} u_j + \p_j u_i).$ 

Based on Theorem \ref{t:t1}, the theory outlined can be physically interpreted as modeling a body primarily made of a Green elastic material with an additional thin, gradient elastic region of thickness $h$, extending from a section of its boundary. More precisely, such a body is modeled by a Green elastic solid $\tilde{\cl B}$ possessing a gradient elastic boundary surface, $\tilde{\cl S}$, with stored surface energy density $U$. The complex interactions of the two distinct three-dimensional regions of the body are encapsulated by the boundary conditions on $\tilde{\cl S}$ appearing in \eqref{eq:linfieldequations}. 

\subsection{Modeling crack fronts under anti-plane shear}

Consider a brittle, infinite plate $\tilde{\cl B}$ under anti-plane shear loading, $\lim_{x^3 \rar \pm \infty} \sigma_{12} = 0$ and $\lim_{x^3 \rar \pm \infty}\sigma_{23} = \sigma$, with a straight crack $\cl C = \{(x^1, x^2, 0) \mid x^1 \in [-\ell, \ell] \}$ of length $2 \ell$, illustrated by Figure \ref{f:2}. For anti-plane shear of the form
\begin{align}
	\bs u(x^1, x^2, x^3) = \su(x^1, x^3) \bs e_2,
\end{align} 
the only nonzero components of the stress are 
\begin{align}
	\sigma_{12} = \frac{\tilde E}{2(1+\tilde \nu)} \su_{,1}, \quad \sigma_{23} = \frac{\tilde E}{2(1+\tilde \nu)} \su_{,3}
\end{align}
By the symmetry of the problem, $\su$ can be taken to be even in $x^1$ and odd in $x^3$, so we will focus only on the strain and stress fields for $x^3 \geq 0$. 

The equations determining the displacement are posited to be \eqref{eq:linfieldequations} on $\tilde{\cl B} = \{ (x^1,x^2,x^3) \mid x^3 \geq 0\}$ with $\tilde{\cl S} = \cl C$, $\bs t = \bs 0$, and $\bs f = \bs 0$. It is physically reasonable to assume that $\tilde E = E$ and $\tilde \nu = \nu$, where $E$ and $\nu$ are used in defining \eqref{eq:surfacestored} through \eqref{eq:abc}. This assumption is rooted in the fact that any region near the crack front, $\tilde{\cl S}$, is made of the same material as the bulk solid, unlike a coating of a different material. If we identify $\ell_s^2 = d^2/12$, with $d$ representing the material's natural inter-particle spacing, the only unspecified parameter is $h$: the thickness of the region near $\tilde{\cl S}$ where small-scale gradient elastic effects become significant.

\begin{figure}[b]
	\centering
	\includegraphics[scale=.75]{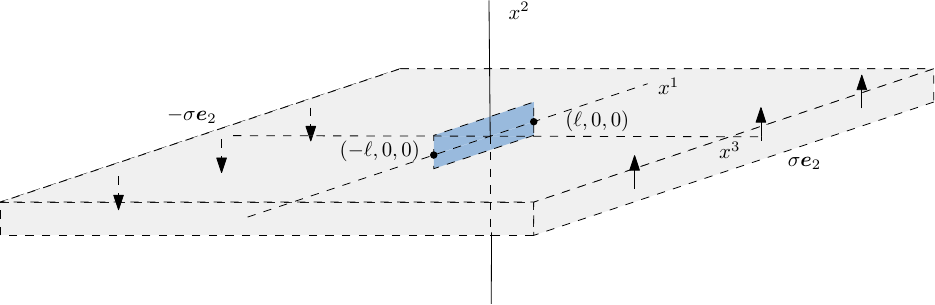}
	\caption{The set-up for the mode-III problem with the crack $\cl C$ appearing in blue.}
	\label{f:2}
\end{figure}

We define dimensionless variables
\begin{align}
	x = \frac{x^1}{\ell}, \quad y = \frac{x^2}{\ell}, \quad z = \frac{x^3}{\ell}, \quad
	\sv(x,z) = \frac{1}{\ell} \Bigl (
	\su(x^1, x^3) - \frac{2(1+\nu)}{E} \sigma x^3
	\Bigr ).
\end{align}
Then the field equations take the dimensionless form 
\begin{align}
	\begin{split}
		&\p_x^2 \sv(x,z) + \p_z^2 \sv(x,z) = 0, \quad z > 0, \\
		&-\sv_z(x,0) = \al \sv_{xx}(x,0) - \beta \sv_{xxxx}(x,0) + \gamma, \quad x \in (-1,1),\\
		&\sv(x,0) = 0, \quad |x| \geq 1, \\
		&\sv_x(\pm 1, 0) = 0, \quad
		\lim_{z \rar \infty} [|\p_x \sv(x,z)| + |\p_z \sv(x,z)|] = 0.
	\end{split}\label{eq:field2}
\end{align} 
The dimensionless parameters $\al, \beta$ and $\gamma$ are given by  
\begin{align}
	\al = \frac{h}{\ell} > 0, \quad \beta = \frac{h \ell_s^2}{\ell^3} > 0, \quad \gamma = \frac{\sigma}{\mu}, \label{eq:size}
\end{align}
with $\beta = \frac{h d^2}{12 \ell^3}$ if we adopt \eqref{eq:specificls}. As discussed in \cite{WaltonNote12, IMRUSch13}, the case $\ell_s = 0$ \emph{does not} lead to a model producing bounded stresses and strains up to the crack tips $x = \pm 1$, i.e., we have
\begin{align}
	\sup_{z > 0} |\nabla_{x,z} \sv(\pm 1, z)| = \infty. 
\end{align}

By using the Hilbert transform (see Section 4 of \cite{Rodriguez2023StrainGradient}), the problem \eqref{eq:field2} can be completely reduced to an integro-differential equation on the boundary for $f := \sv|_{z = 0}$,
\begin{align}
\begin{split}
	&\beta f''''(x) - \alpha f''(x) + \cl H f'(x)= \gamma , \quad x \in (-1,1), \\ 
	&f(\pm 1) = f'(\pm 1) = 0.
\end{split}	\label{eq:integrodiff} 
\end{align}
Since $\beta > 0$, the main result of Section 4 in \cite{Rodriguez2023StrainGradient} directly applies to \eqref{eq:integrodiff}, showing that our model for fracture rooted in \cite{Rodriguez2023StrainGradient} and utilizing $U$ from this work generates stresses and strains that remain bounded up to the crack tips.
\begin{thm}[Theorem 4.4, \cite{Rodriguez2023StrainGradient}]\label{t:bounded}
	There exists $C > 0$ depending on $\al$ and $\beta$ such that the following hold. There exists a unique classical solution $f \in C^4([-1,1])$ to \eqref{eq:integrodiff}, and $f$ satisfies 
	\begin{align}
		\| f \|_{C^4([-1,1])} \leq C |\gamma|. \label{eq:festimate}
	\end{align}
	Moreover, the dimensionless displacement field $\sv(x,z) = \int_{-\infty}^\infty P_z(x-s) f(s) ds$, where $P_z(\cdot)$ is the Poisson kernel for the upper half plane, produces bounded stresses and strains up to the crack tips: 
	\begin{align}
		\| \sv \|_{C^1(\{ z \geq 0 \})} \leq C |\gamma|. \label{eq:westimate}
	\end{align} 
\end{thm}

\bibliographystyle{plain}
\bibliography{researchbibmech}
\bigskip

\centerline{\scshape C. Rodriguez}
\smallskip
{\footnotesize
	\centerline{Department of Mathematics, University of North Carolina}
	
	\centerline{Chapel Hill, NC 27599, USA}
	
	\centerline{\email{crodrig@email.unc.edu}}
}

\end{document}